\documentclass[journal, 10pt]{IEEEtran}

\usepackage{graphicx,color}
\usepackage{cite}
\usepackage{setspace} 
\usepackage{amsmath,amsthm,amssymb}
\usepackage{multirow}
\usepackage{hhline}
\usepackage{epsfig}
\usepackage{epstopdf}
\usepackage{verbatim}
\usepackage{algorithm}
\usepackage{algorithmicx}
\usepackage{algpseudocode}
\usepackage{cases}
\usepackage{array}

\newtheorem{theorem}{Theorem}[section]
\newtheorem{proposition}[theorem]{Proposition}

\def\IEEEproof{\textbf{Proof:}~}
\def\endIEEEproof{\textbf{\hfill$\Box$}}

\usepackage{forloop}
\newcounter{ct}

\newenvironment{remark}[1][Remark:]{\begin{trivlist}
\item[\hskip \labelsep {\bfseries #1}]}{\end{trivlist}}	

\begin{document}

\title{Jammer-Assisted Resource Allocation\\ in Secure OFDMA with Untrusted Users\thanks{Copyright (c) 2013 IEEE. Personal use of this material is permitted. However, permission to use this material for any other purposes must be obtained from the IEEE by sending a request to pubs-permissions@ieee.org.}
\thanks{
This work has been supported by the Department of Science and Technology (DST) under Grant SB/S3/EECE/0248/2014. The associate editor coordinating the review of this paper and approving it for publication was xxxxxxxxxxxxx.}} 
\author{Ravikant Saini, 
 Abhishek Jindal\thanks{R. Saini and A. Jindal are with the Bharti School of Telecom, IIT Delhi, New Delhi, India. (e-mail: ravikant.saini@dbst.iitd.ac.in; abhishek.jindal@dbst.iitd.ac.in)}, and Swades De 
 \thanks{S. De is with the Department of Electrical Engineering and Bharti School of Telecom, IIT Delhi, New Delhi, India. (e-mail: swadesd@ee.iitd.ac.in).}
\thanks{Color  versions  of  one  or  more  of  the  figures  in  this  paper  are  available online at http://ieeexplore.ieee.org.}
\thanks{Digital  Object Identifier  xx.xxxx/TIFS.2016.xxxxxxx.}
}

\maketitle

\begin{abstract}
In this paper, we consider the problem of resource allocation in an OFDMA system with single source and $M$ untrusted users in presence of a friendly jammer.
The jammer is used to improve either the weighted sum secure rate or the overall system fairness.
The formulated optimization problem in both the cases is a Mixed Integer Non-linear Programming (MINLP) problem, belonging to the class of NP-hard.
In the sum secure rate maximization scenario, we decouple the problem and first obtain the subcarrier allocation at source and the decision for jammer power utilization on a per-subcarrier basis. 
Then we do joint source and jammer power allocation using primal decomposition and alternating optimization framework.  
Next we consider fair resource allocation by introducing a novel concept of subcarrier snatching with the help of jammer.
We propose two schemes for jammer power utilization, called proactively fair allocation (PFA) and on-demand allocation (ODA).
PFA considers equitable distribution of jammer power among the subcarriers, while ODA distributes jammer power based on the user demand.
In both cases of jammer usage, we also present suboptimal solutions that solve the power allocation at a highly reduced complexity. 
Asymptotically optimal solutions are derived to benchmark optimality of the proposed schemes.
We compare the performance of our proposed schemes with equal power allocation at source and jammer.
Our simulation results demonstrate that the jammer can indeed help in improving either the sum secure rate or the overall system fairness.
\end{abstract}

{\IEEEkeywords Secure OFDMA, friendly jammer, subcarrier snatching, rate maximization, max-min fairness}

\section{Introduction}
\label{sec_introduction}
Rapid deployment of wireless communication systems has always been engraved with the issue of security of the transmitted data. 
The basic reason is the broadcast nature of transmission which makes the signals vulnerable to tapping by malicious users \cite{amitav_TCST_2014}. 
Security of the transmitted signal is generally considered as a responsibility of the higher layers which employ cryptographic techniques. This strategy relies on the basic assumption that the enciphering system is unbreakable by the malicious users \cite{Shiu_WC_2011}. 
With growing computing capabilities, such measures prove to be insufficient and it has
motivated the research community to explore security at the physical layer.
Physical layer security finds its basis from independence of the wireless communication channels and has a low implementation complexity \cite{Bloch_book_2011}.
This added layer of security is considered as the strictest kind of security, not requiring even any kind of key exchange \cite{amitav_TCST_2014, Shiu_WC_2011}.

Orthogonal frequency division multiple access (OFDMA) is a potential physical layer technology for the next generation access networks such as WiMAX, LTE, and beyond. Hence, the study of physical layer security in OFDMA has gained considerable attention in recent years \cite{Haohao_ICC_2013, Karachontzitis_TIFS_2015,Derrick_TVT_2012,  Derrick_TWC_2011, Munnujahan_ICC_2013, Wang_eurasip_2013, Jorswieck2008, Xiaowei_TIFS_2011, Jorswieck_SAC_2013}.
The subcarrier allocation and power optimization in multi-node secure OFDMA system has been studied in two cases: in one scenario transmitter assumes all users to be trusted and there exists an external eavesdropper which tries to decode the data of trusted users \cite{Haohao_ICC_2013, Karachontzitis_TIFS_2015, Derrick_TVT_2012,  Derrick_TWC_2011, Munnujahan_ICC_2013}, while in another scenario there is a mutual distrust among the users and a particular user may demand the source to transmit its data considering all other users as the potential eavesdroppers \cite{Xiaowei_TIFS_2011, Jorswieck2008, Jorswieck_SAC_2013}. 

In \emph{trusted users case}, the authors in \cite{Karachontzitis_TIFS_2015} considered maximizing the minimum of weighted sum secure rate of all the users. They proposed a complex MILP based solution and less complex suboptimal schemes. However, the issue of resource scarcity of a user that is very near to the eavesdropper was not highlighted. 
Maximization of energy efficiency with multi-antenna source, eavesdropper, and  single antenna users was studied in \cite{Derrick_TVT_2012} with bounds on per-user tolerable secrecy outage probability.
The average secrecy outage capacity maximization in the same setup with a multi-antenna Decode and Forward relay was studied in \cite{Derrick_TWC_2011}.
The authors in \cite{Derrick_TWC_2015} studied secure rate maximization under fixed quality of service (QoS) constraints for secure communication among multiple source-destination pairs in the presence of multi-antenna external eavesdropper, with the help of multiple single antenna Amplify and Forward (AF) relays.
An extension of the above proposed schemes to multiple eavesdroppers case may be straight forward, but with untrusted users the problems present an interesting challenge with multiple unexplored facets. 
Looking at another domain of untrust, subcarrier assignment and power allocation problem was studied in \cite{Wang_eurasip_2013} in an AF untrusted relay aided secure communication among multiple source-destination pairs.

In an effort to tackle the \emph{untrusted users}, the authors in \cite{Jorswieck2008} proposed to allocate a subcarrier to its best gain user in a two-user OFDMA system and presented the optimal power allocation. 
The scheme can be trivially extended to more than two users for sum rate maximization.
In a scenario having heterogeneous demand of resources by users, the authors in \cite{Xiaowei_TIFS_2011} proposed a joint subcarrier and power allocation policy for two classes of users: secure users demanding a fixed secure rate and normal users which are served with best effort traffic. 
In the cognitive radio domain, secure communication over a single carrier, between multi-antenna secondary transmitter and a fixed user considering other secondary and primary users as eavesdroppers was studied in \cite{Derrick_TVT_2015}. 
Similarly, precoder design to maximize the sum secrecy rate to achieve confidential broadcast to users with multiple antennas was studied in \cite{Yang_TC_2014, Geraci_TC_2012}.
The authors in \cite{Jorswieck_SAC_2013} solved the resource allocation problem among multiple source-destination pairs in a relay-aided scenario with untrusted users.
However, in secure OFDMA fair resource distribution among untrusted users poses new challenges which have not yet been investigated. 

\subsection{Motivation}
The study in \cite{Xiaowei_TIFS_2011} raised the feasibility issue of resource allocation problem because of the channel conditions, i.e., some users may not achieve the required secure rate due to tapping by untrusted users. This motivates us to investigate if the secure rate over a subcarrier could be any how improved.
The strategy proposed in \cite{Munnujahan_ICC_2013} utilizes interference to improve the secure rate of a user. But this scheme cannot be used in multi-user untrusted scenario because it relies on a strong assumption that the jammer affects the eavesdropper only. 
It was shown in \cite{SGOEL_TWC_2008} that, for single user multi-antenna system, secrecy rate can be improved by a multi-antenna source by jointly transmitting message in the range space and interference in the null space of the main channel. Since this solution is difficult to apply for single antenna systems, the authors in \cite{XTANG_TIT_2011} showed the possibility of achieving a positive rate even when eavesdropper's channel is stronger compared to main channel by using jamming power control only.
The studies in \cite{Munnujahan_ICC_2013, SGOEL_TWC_2008, XTANG_TIT_2011} are based on single user system. The implications of utilizing interference in multi-user scenario are yet to be investigated.    
With regard to fair resource allocation without the help of interference, the strategy proposed in \cite{Karachontzitis_TIFS_2015}, which balances the secure rate among users, can limit the maximum achievable rate of the system because of a poor user facing strong eavesdropping by other users. This leads to system resource wastage.
Effectively utilizing jammer power control in multiple untrusted user scenario for sum rate maximization or fair resource allocation raises challenges which to the best of our knowledge has not yet been investigated in the literature. 

In this paper, we intend to explore the role of the jammer power in alleviating the following issues which plague single antenna secure OFDMA system with $M$ untrusted users:
\begin{itemize}
\item Secure rate achievable by the users can be very low as compared to the single external eavesdropper case, because now there are $M-1$ wiretappers instead of one; 
\item A large number of users can starve for subcarriers, because a group of users with very good channel gains may prohibit secure communication to the other users.
\end{itemize}

Assuming the jammer to be a node affecting all the users, we intend to use the jammer power for individual and independent jamming over the subcarriers.
The resource allocation problem in presence of jammer is a complex Mixed Integer Non Linear Programming (MINLP) problem belonging to the class of NP-hard \cite{Di_Yuan_TVT_2013,Y_F_LIU_TSP_2014}. 
Due to the requirement of decision on subcarrier allocation at the source as well as to use jammer power over a subcarrier, the problem exhibits combinatorial nature having exponential complexity with number of subcarriers \cite{MTAO_TWC_2008}.
Hence, instead of attempting a global optimal solution, we study the resource allocation problem to improve either the sum secure rate or the overall system fairness. 
While extending our preliminary work in \cite{ravikant_ICC_2015}, where we presented two suboptimal solutions based on sequential source and jammer power allocation, 
we conduct a deeper study on the behavior of secure rate with jammer power and 
solve the joint source and jammer power optimization problem. 
For max-min fair resource allocation, two novel methodologies of jammer power utilization are introduced.
We also present the asymptotic analysis of the algorithms with $P_S \to \infty$ and $P_J \to \infty$, and discuss the computational complexity of all the proposed schemes under both the usages of jammer.

\subsection{Contribution}
Our contributions can be summarized as answers to the following two questions:

(1) \emph{How to improve sum secure rate in OFDMA systems with untrusted users?}
To address this question, we consider the possibility of using jammer power over a subcarrier to improve the secure rate beyond what can be achieved by source power only.
The features of the proposed solution are as follows:
\begin{itemize}
\item We obtain the constraints of secure rate improvement over a subcarrier and show that the secure rate in this constrained domain is a quasi-concave function of jammer power, thereby offering  a unique maxima.
\item The jammer aided approach introduces a new challenge - referred as SNR reordering in the rest of the paper - which makes the problem combinatorial even after subcarrier allocation and jammer utilization decision. 
We introduce the concept of constrained-jamming by developing jammer power bounds to handle this challenge.
\item For the known subcarrier allocation and jammer decision, the joint source and jammer power allocation problem is solved optimally using \emph{primal decomposition (PD)}\cite{sboyd_decomposition} and \emph{alternating optimization (AO)}\cite{sboyd_alternating} techniques. 
\item We analyze the complexity of the proposed algorithm and show its convergence in finite steps.
\item We also propose less complex solution which allows a trade-off between performance and complexity.
\item Asymptotically optimal solution is derived to assess optimality of the proposed scheme.
\end{itemize}

(2) \emph{How to remove subcarrier scarcity and have fair distribution?}
In secure OFDMA with untrusted users and in absence of jammer, the best gain user over a subcarrier is the strongest eavesdropper for all other users. Thus, a subcarrier should normally be given to its best gain user, which may cause some users to starve for channel resources. 
To attain a better fairness, we propose to take away a subcarrier from its best gain user and allocate it to another user with poor channel gain with the help of jammer.
The features of the proposed solution are:
\begin{itemize}
\item Secure rate over the snatched subcarrier is positive when the jammer power is above a certain threshold and attains a maximum value at an optimal jammer power.
\item The conventional max-min fairness algorithm cannot be employed in secure OFDMA because a user with poor channel gains over all the subcarriers may force the algorithm into a deadlock. The proposed max-min fair scheme offers a graceful exit with such users. 
\item Two variants of max-min fairness algorithms are proposed where the jammer power can be utilized either by preserving it  for possible snatching over each subcarrier, or by allocating it based on the demand of snatching user.
\item The power allocation is done by PD and AO process while the issue of signal to noise ratio (SNR) reordering is handled similarly as in sum rate maximization case.
\item Low complexity solutions have been presented for both the variants of max-min fairness algorithms. 
\item We obtain an asymptotic upper bound of the fairness achievable by our proposed scheme. 
\end{itemize}

The remainder of the paper is organized as follows: The system model is introduced in Section \ref{sec:model}. The sum secure rate maximization problem is studied in Section \ref{sec_sum_rate_imp}, followed by fair resource allocation study in Section \ref{sec_max_min_fairness}. Asymptotic analysis has been presented in Section \ref{sec_asymp_anlysis}. Simulation results are discussed in Section \ref{sec_results}.  Section \ref{sec_conclusion} concludes the paper.

\section{System model}\label{sec:model}
We consider the downlink of an OFDMA system with a single source (base station), a friendly jammer, $M$ untrusted mobile users (MUs) which are randomly distributed in the cell coverage area, and $N$ subcarriers. The jammer, which is considered as another node in the network (e.g., in LTE-Advanced it could be an idle relay node), is controlled by the source to help improve the overall system performance. 
Jammer is assumed to be capable of collecting channel state information (CSI) in the uplink and sending jamming signal which is unknown to users on the downlink.
Since all users request secure communication with the source, they share their CSI with the source as well as the jammer.
In the present study we assume the jammer to be located randomly in the cell area.

Source-to-MUs and jammer-to-MUs channels are considered to experience slow and frequency-flat fading, such that the channel parameters remain constant over a frame duration but vary randomly from one frame to another.
Perfect CSI of source-to-MU and jammer-to-MU channel pairs for all the MUs is assumed to be available at source \cite{Haohao_ICC_2013, Xiaowei_TIFS_2011, Derrick_TVT_2012, Derrick_TWC_2011, Karachontzitis_TIFS_2015, Munnujahan_ICC_2013, Wang_eurasip_2013, Jorswieck2008}.
Source utilizes this information for subcarrier allocation, and source and jammer power allocation. 
We consider that a subcarrier is exclusively allocated to one user only, which has been proved to be optimal for sum secure rate maximization \cite{Jorswieck2008}.

\begin{figure}[!htb]
\centering
\epsfig{file=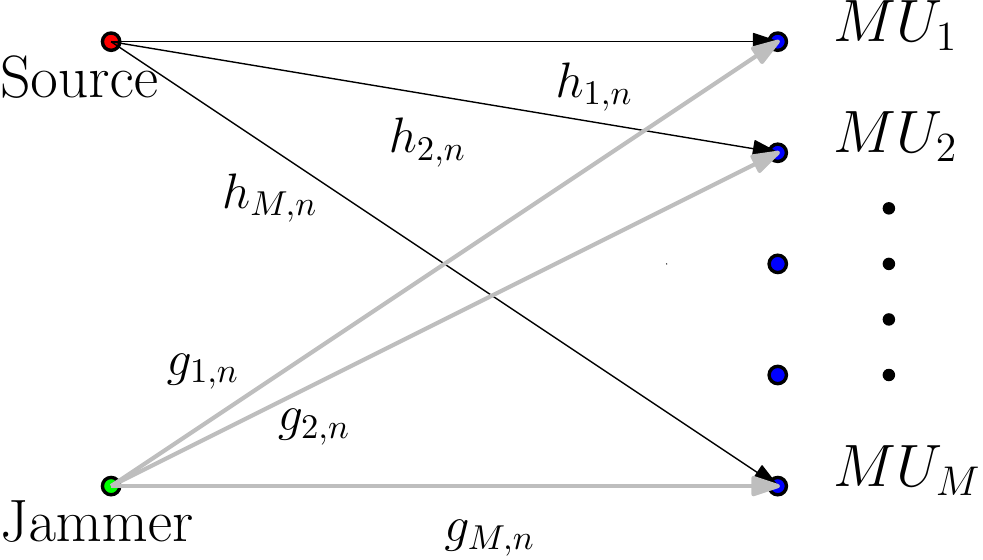, height=1.2in}
\caption{System model}
\label{fig:system_model}
\end{figure}

The secure OFDMA system with $M$ untrusted users is a multiple eavesdropper scenario, where for each main user there exist $(M-1)$ eavesdroppers.
Out of these $(M-1)$ eavesdroppers, the strongest one is considered as the equivalent eavesdropper (hereafter referred as eavesdropper). 
The secure rate of the main user over a subcarrier is defined as the non-negative capacity difference between the main user and the eavesdropper. Let $h_{i,n}$ be the source to $i$th user channel coefficient and $g_{i,n}$ be the jammer to $i$th user channel coefficient over subcarrier $n$ as shown in Fig. \ref{fig:system_model}. Then, the secure rate of user $m$ on $n$th subcarrier is given as \cite{Jorswieck2008, Xiaowei_TIFS_2011}:
\begin{align}\label{secure_rate_definition}
R_{m,n} & = \left[ \log_2 \left( 1+\frac{P_{s_n}|h_{m,n}|^2}{\sigma^2 + \pi_{j_n} P_{j_n} |g_{m,n}|^2 } \right) \right. \nonumber \\
 & \hspace{-0.12in} - \left. \max\limits_{e \in \{1,2,\cdots M\} \setminus m} \log_2 \left( 1+\frac{P_{s_n} |h_{e,n}|^2}{\sigma^2+ \pi_{j_n} P_{j_n} |g_{e,n}|^2 } \right)  \right]^+
\end{align}
where $(x)^+=\max(x,0)$, $\sigma^2$ is the AWGN noise variance,  $P_{s_n}$ and $P_{j_n}$ are respectively source and jammer powers over subcarrier $n$. $\pi_{j_n} \in \{0, 1\}$ is an indicator of absence or presence of jammer power on subcarrier $n$ such that $\pi_{j_n}=0 \implies P_{j_n}=0$, and $e$ is the eavesdropper. The non-linearity of secure rate in source and jammer powers along with the $\max$ operator complicate the optimal resource allocation which will be discussed in the Section \ref{sec_sum_rate_imp} and \ref{sec_max_min_fairness}.

\emph{Example:} 
Let us consider a symbolic 3 user 5 subcarrier OFDMA system, with the respective channel gains $|h_{m,n}|$ and $|g_{m,n}|$ given in Tables \ref{src_gain} and \ref{jammer_gain}. This system will be used to illustrate the key concepts developed further in the paper. \hspace{0.2in}\textbf{$\Box$}

\section{Resource allocation for sum secure rate maximization}
\label{sec_sum_rate_imp}
In this section we discuss the joint source and jammer resource optimization problem for weighted sum secure rate maximization. The problem can be stated as:
\begin{align}\label{opt_prob_jra_rate_max}
& \underset{P_{s_n}, P_{j_n}, \pi_{m,n}, \pi_{j_n}} {\text{maximize}} \sum_{m=1}^M \sum_{n=1}^N \ w_m \pi_{m,n}  R_{m,n} \nonumber \\
& \text{subject to}  \nonumber \\ 
&C_{1,1}: \sum_{n=1}^N P_{s_n} \leq P_S, \qquad \quad C_{1,2}: \sum_{n=1}^N P_{j_n} \leq P_J, \nonumber \\
&C_{1,3}: \pi_{j_n} \in \{0,1\}  \text{ } \forall n, \qquad C_{1,4}: \pi_{m,n} \in \{0,1\} \text{ } \forall m, n,\nonumber \\
&C_{1,5}: \sum_{m=1}^M \pi_{m,n} \leq 1 \text{ } \forall n, \quad C_{1,6}:P_{s_n}\ge 0, P_{j_n}\ge 0 \text{ } \forall n
\end{align} 
where $\pi_{m,n}$ is a binary allocation variable to indicate whether  subcarrier $n$ is given to user $m$ or not, $w_m$ is the priority weight allocated by the higher layers to user $m$, and 
$P_S$ and $P_J$ are source and jammer power budgets, respectively.
$C_{1,1}$ and $C_{1,2}$ are budget constraints, $C_{1,3}$ is jammer allocation constraint, $C_{1,4}$ and $C_{1,5}$ are subcarrier allocation constraints, and $C_{1,6}$ denotes source and jammer power boundary constraints.

The optimization problem in (\ref{opt_prob_jra_rate_max}) has total four variables per subcarrier:  $\pi_{j_n}$, $\pi_{m,n}$ as binary variables; and $P_{s_n}$ and $P_{j_n}$ as continuous variables. 
Since the problem is a non-convex combinatorial problem belonging to the class of NP-hard, there is no polynomial time optimal solution possible \cite{Di_Yuan_TVT_2013,Y_F_LIU_TSP_2014}. We tackle the problem by breaking it in parts and attempt to find a near-optimal solution, which approaches asymptotically optimal solution (discussed in Section \ref{sec_asymp_anlysis}) as $P_S$ and $P_J$ increases.
For this, first we perform subcarrier allocation without considering the jammer.
Next we distinguish those subcarriers over which jammer can improve the secure rate, 
and finally we complete the joint power allocation.

%
\newcolumntype{C}{>{\centering\arraybackslash}p{3em}}
\newcolumntype{L}[1]{>{\raggedright\let\newline\\\arraybackslash}p{#1}}
\begin{table}[!htb]
\caption{Source-users channel gains $|h_{m,n}|$}\label{src_gain}
\centering
{\scriptsize
\begin{tabular}{|L{.5cm} |L{1cm}| L{1cm}| L{1cm}| L{1cm}| L{1cm}|}
\hline
{} &    {$c_1$} & {$c_2$} & {$c_3$} & {$c_4$} & {$c_5$} \\
\hline
{$u_1$} & {1.1027} & {0.3856} & {0.6719} & {1.2101} & {0.7043} \\
\hline
{$u_2$} & {0.7423} & {1.0735} & {0.6558} & {1.0006} & {0.8943} \\
\hline
{$u_3$} & {0.7554} & {1.4772} & {0.2498} & {1.3572} & {3.5391} \\
\hline
\end{tabular}
}
\end{table} 

%
\begin{table}[!htb]
\caption{Jammer-users channel gains $|g_{m,n}|$}\label{jammer_gain}
\centering
{\scriptsize
\begin{tabular}{|L{.5cm} |L{1cm}| L{1cm}| L{1cm}| L{1cm}| L{1cm}|}
\hline
{} &    {$c_1$} & {$c_2$} & {$c_3$} & {$c_4$} & {$c_5$} \\
\hline
{$u_1$} & {3.3624} & {6.0713} & {3.4125} & {3.0584} & {0.4987} \\
\hline
{$u_2$} & {8.1741} & {7.0607} & {4.1047} & {0.9860} & {1.6860} \\
\hline
{$u_3$} & {0.9028} & {2.0636} & {0.5605} & {3.0277} & {4.5346} \\
\hline
\end{tabular}
}
\end{table}

\subsection{Subcarrier allocation at source} 
\label{subsec_subc_alloc_src}
In presence of jammer, subcarrier allocation at source, i.e., $\pi_{m,n}$ decision is difficult. 
In secure OFDMA with untrusted users a subcarrier can only be given to the user having maximum SNR over the subcarrier. Any change in $P_{j_n}$ may force the decision of subcarrier allocation to change because of the possible reordering of updated SNRs (cf. (\ref{secure_rate_definition})). Since $P_{s_n}$ is in numerator, any change in $P_{s_n}$ does not affect the SNR ordering.
This makes the problem combinatorial, as we may need to check subcarrier allocation for every update of $P_{j_n}$.
In order to bypass this step, we initially assume that the jammer is not present.
In absence of jammer $\pi_{j_n}=0$ $\forall$ $n$, and the secure rate 
definition in (\ref{secure_rate_definition}) changes to  
\begin{align}\label{secure_rate_definition_without_jammer}
R_{m,n}|_{_{\pi_{j_n}=0}} & = \left[ \log_2 \left( 1+\frac{P_{s_n}|h_{m,n}|^2}{\sigma^2} \right) \right. \nonumber \\
 & \hspace{-0.12in} - \left. \max\limits_{e \in \{1,2,\cdots M\} \setminus m} \log_2 \left( 1+\frac{P_{s_n} |h_{e,n}|^2}{\sigma^2} \right)  \right]^+.
\end{align}
As observed in (\ref{secure_rate_definition_without_jammer}), $|h_{m,n}|$ $>$ $\max\limits_{e \in \{1,2,\cdots M\} \setminus m} |h_{e,n}|$ is required to have positive secure rate over subcarrier $n$, irrespective of $P_{s_n}$. 
Thus, the subcarrier allocation policy, allocating a subcarrier to its best gain user can be stated as: 
\begin{align}\label{subcarrier_alloc_policy}
\pi_{m,n}=
\begin{cases}
1, & \text{if $|h_{m,n}|=\max\limits_{e \in \{1,2,\cdots M\}}|h_{e,n}|$}\\
0, & \text{otherwise.}
\end{cases} 
\end{align}

\begin{remark}
Note that the subcarrier allocation (\ref{subcarrier_alloc_policy}) depends only on users' gains. It is indifferent to users' priority imposed through weights which will play their role in power allocation.  
\end{remark}

\subsection{Subcarrier allocation at jammer and jammer power bounds}\label{subsec_subc_alloc_jammer_bounds}
In the presence of jammer, even after subcarrier allocation at source is done, the problem is still combinatorial due to  $\pi_{j_n}$. Any change in $P_{j_n}$ may even jeopardize the earlier decision on $\pi_{m,n}$ as described in the previous section. In-order to solve this issue, we first introduce the concept of rate improvement, which will help us decide $\pi_{j_n}$ retaining the decision on $\pi_{m,n}$. 

\subsubsection{Selective jamming for secure rate improvement}\label{subsec_sel_jamm_rate_imp}
Since jammer affects all the users, it appears that using jammer over a subcarrier may degrade the secure rate.
But it is interesting to note that, with jammer power the secure rate can be improved beyond what is achieved without jammer.
For the proof of concept, let us consider a simple OFDMA scenario having four nodes: source, jammer, and two users $m$ and $e$.
Let a subcarrier $n$ be allocated to user $m$, and $e$ be the eavesdropper.
The following proposition describes the possibility of secure rate improvement and the existence of optimal jammer power achieving maximum secure rate over subcarrier $n$.

\begin{proposition}\label{prop_sri_jammer_rate_imp_constraints}
The secure rate over a subcarrier $n$ having $\lvert{h_{m,n}}|>|{h_{e,n}}|$ can be improved if $|{g_{e,n}}|>|{g_{m,n}}|$ and the source and jammer powers 
$P_{s_n}$ and $P_{j_n}$ are constrained as
\begin{align}\label{eqn_pjn_rate_imp_cond}
P_{s_n}&>
\begin{cases}
 \frac{ \sigma^2 \left( |g_{m,n}|^2 |h_{m,n}|^2 - |g_{e,n}|^2 |h_{e,n}|^2 \right) }{ \left( |g_{e,n}|^2 - |g_{m,n}|^2 \right) |h_{m,n}|^2 |h_{e,n}|^2 },  & \text{if } \frac{|g_{m,n}||h_{m,n}|}{|g_{e,n}||h_{e,n}|} > 1 \\
0, & \text{otherwise}
\end{cases} \nonumber\\
\text{and }
& P_{j_n} <  \frac{P_{s_n} \alpha_n+\sigma^2 \beta_n} {|g_{m,n}|^2 |g_{e,n}|^2 \left( |h_{m,n}|^2- |h_{e,n}|^2 \right) } \triangleq P_{j_{n}}^{th_{i}}
\end{align}
where
$\alpha_n = \left( |g_{e,n}|^2 - |g_{m,n}|^2 \right) |h_{m,n}|^2 |h_{e,n}|^2$ and $\beta_n = \left( |g_{e,n}|^2 |h_{e,n}|^2 - |g_{m,n}|^2 |h_{m,n}|^2 \right)$.

In the constrained domain of rate improvement, the rate is a quasi-concave function of $P_{j_n}$ having a unique maxima. 
\end{proposition}
\begin{proof}
See Appendix \ref{sec_appendix_proposition_rate_imp}.
\end{proof}

\emph{Example (continued):} 
Observing the source gains in Table \ref{src_gain}, the best gain users of the subcarriers are $u_1, u_3, u_1, u_3, u_3$, respectively. 
Users $u_3, u_2, u_2, u_1, u_2$ are their corresponding eavesdroppers. 
Subcarriers $c_2, c_3, c_4$ satisfy the condition of secure rate improvement, i.e., $|g_{e,n}|>|g_{m,n}|$.
Considering the unit of transmit power is in Watt, let $P_S = 10$, $P_J=10$, and $\sigma^2=1$. 
Source power thresholds for the three subcarriers are  $P_{s_n}^{th_i} = 0.0, 0.0, 6.3263$, respectively. 
Assuming equal source power allocation over all  subcarriers ($P_{s_n} = 2$) and comparing $P_{s_n}$ with $P_{s_n}^{th_i}$, subcarriers $c_2, c_3$ can be utilized for secure rate improvement while $c_4$ cannot be used.
The corresponding jammer power thresholds for $c_2, c_3$ are respectively found as $P_{j_n}^{th_i} = 1.2693, 0.9560$. 
Variation of SNRs of the users and the secure rate $R_{3,2}$ of user $u_3$ on subcarrier $c_2$ with jammer power $P_{j_2}$ are presented in Table \ref{carreir_2}. 
As observed, $R_{3,2}$ with $P_{j_2}>0$ is higher compared to the value when $P_{j_2}=0$, till $P_{j_2}<P_{j_2}^{th_i} = 1.2693$. $R_{3,2}$ has a maxima between $P_{j_2}=0$ and $P_{j_2}=P_{j_2}^{th_i}$, at $P_{j_2}^o = 0.1027$.\hfill\textbf{$\Box$}

\begin{table}[!htb]
\caption{Users' SNRs and secure rate $R_{3,2}$ versus $P_{j_2}$}\label{carreir_2}
\centering
{\scriptsize
\begin{tabular}{|L{1cm} |L{1cm}| L{1cm}| L{1cm}| L{1cm}| L{1cm}| L{1cm}|}
\hline
{$P_{j_2} \to $} &    {$0.0$} & {$0.1$} & {$0.2$} & {$1.2$} &  {$1.3$}  \\
\hline
{$SNR_{1}$} & {0.1487} & {0.0317} & {0.0178} & {0.0033} & {0.0030} \\
\hline
{$SNR_{2}$} & {1.1524} & {0.1925} & {0.1050} & {0.0189} & {0.0175} \\
\hline
{$SNR_{3}$} & {2.1821} & {1.5304} & {1.1784} & {0.3571} & {0.3339} \\
\hline
{$R_{3,2}$} & {0.6988} & {1.5518} & {1.4720} & {0.7239} & {0.6882}\\
\hline
\end{tabular}
}
\end{table}

Based on the results described in Proposition \ref{prop_sri_jammer_rate_imp_constraints}, we create a prospective set of subcarriers, over which the jammer power can be used for secure rate improvement.
The subcarrier allocation policy at jammer can be summarized as: 
\begin{align}\label{subcarrier_alloc_jammer}
\pi_{j_n}=
\begin{cases}
1, & \text{if $|g_{e,n}|>|g_{m,n}|$}\\
0, & \text{otherwise.}
\end{cases} 
\end{align}

While extending the result of Proposition \ref{prop_sri_jammer_rate_imp_constraints} to $M>2$, there is an inherent challenge associated with the allocation of $P_{j_n}$, which we refer as SNR reordering. Without loss of generality, let the channel gains from source-to-MUs be sorted as $|h_{1,n}|>|h_{2,n}|>\cdots>|h_{M,n}|$ over a subcarrier $n$ such that user 1 is assigned the subcarrier and user 2 is the eavesdropper. While optimizing source and jammer powers jointly, any update in $P_{s_n}$ does not disturb the resource assignments, but as $P_{j_n}$ is updated it raises the following two concerns:\\
(i) User 1 may not have the maximum SNR.\\
(ii) User 2 may not remain the corresponding eavesdropper.\\
Because of this SNR reordering challenge, for every $P_{j_n}$ update a new main user and the corresponding eavesdropper are to be determined. 
This is the problem introduced by the $\max$ operator appearing in the rate definition  
(\ref{secure_rate_definition}), which makes the joint source and jammer power allocation a tedious task even after $\{\pi_{m,n}\}$ and $\{\pi_{j_n}\}$ allocation.
In order to address this challenge we develop a strategy as described below.

\subsubsection{Bounds on jammer power to avoid SNR reordering}\label{subsubsec_ri_bounds}
In order to handle the $\max$ operator, we enforce certain constraints over jammer power so as to retain the same main user and the same eavesdropper throughout the jammer power allocation.
In order to retain the eavesdropper, the rate of user $e$ should be larger 
than all other possible eavesdroppers, i.e., 
\begin{align}\label{eqn_jra_SPA_pjn_eve_bound}
 \frac{P_{s_n} |h_{e,n}|^2}{\sigma^2 + P_{j_{n}}|g_{e,n}|^2} > \frac{P_{s_n} |h_{k,n}|^2}{\sigma^2 + P_{j_{n}}|g_{k,n}|^2}, \text{ }k \in \{1,\cdot \cdot M \} \setminus \{ m,e \}.
\end{align}
Similarly, to preserve the main user we need to have: 
\begin{equation}\label{eqn_jra_SPA_pjn_main_eve_bound}
 \frac{P_{s_n} |h_{m,n}|^2}{\sigma^2 + P_{j_{n}}|g_{m,n}|^2} > \frac{P_{s_n} |h_{e,n}|^2}{\sigma^2 + P_{j_{n}}|g_{e,n}|^2}.
\end{equation}
The constraints in (\ref{eqn_jra_SPA_pjn_eve_bound}) and (\ref{eqn_jra_SPA_pjn_main_eve_bound}) depend on channel conditions only, which result in lower and upper bounds over $P_{j_n}$ such that $P_{j_{k,n}}^u>P_{j_n}>P_{j_{k,n}}^l$. Since (\ref{eqn_pjn_rate_imp_cond}) enforces another upper bound on $P_{j_n}$, the final lower and upper bounds  are given as 
\begin{align}\label{jammer_power_bounds}
P_{j_{n}}^l = \underset{k} \max \{P_{j_{k,n}}^l\}; \quad P_{j_{n}}^u = \underset{k} \min \{P_{j_{n}}^{th_i}, P_{j_{k,n}}^u\}
\end{align}
\begin{remark}
Since $|h_{m,n}|$ $>$ $|h_{e,n}|$ $>$ $|h_{k,n}|$, hence (\ref{eqn_jra_SPA_pjn_eve_bound}) and (\ref{eqn_jra_SPA_pjn_main_eve_bound}) may result only in upper bounds, and lower bounds are zero. In case $P_{j_n}^l>P_{j_{n}}^u$, there is no feasible region for $P_{j_n}$ and subcarrier $n$ cannot be considered for rate improvement.
\end{remark}

\emph{
Example (continued):} 
The SNR reordering issue can be explained with the help of jammer power variation over subcarrier $c_3$. 
The respective SNRs of the users with $P_{j_3}$ variation are presented in Table \ref{carrier_3}.
From $P_{j_3}= 0$ to $0.4$, $u_1$ and $u_2$ are respectively the main user and the eavesdropper. At $P_{j_3} = 0.5$, the eavesdropper changes from $u_2$ to $u_3$ and this reduces the secure rate $R_{1,3}$ from $0.0616$ to $0.0315$.  When $P_{j_3}$ increases to $0.7$, even the main user changes from $u_1$ to $u_3$, and now $u_1$ is the corresponding eavesdropper. Now, the secure rate of $u_3$ is $0.0048$, while that of $u_1$ is $0.0$. Though the jammer power threshold $P_{j_n}^{th_i}$ guarantees that, if $P_{j_n}<P_{j_n}^{th_i}$, the secure rate remains greater than that without jammer ($P_{j_n}=0$), this assertion relies on the basic assumption that the main user and the eavesdropper do not change. The jammer power threshold, optimum jammer power, and its upper bound to avoid SNR re-ordering on $c_3$ are respectively $P_{j_3}^{th_i} =  0.9560$, $P_{j_3}^o = 0.0808$, and $P_{j_{3,3}}^{u} = 0.4013$.
As soon as 
$P_{j_{3}}>P_{j_{3}}^{u} = \min\{P_{j_3}^{th_i},P_{j_{3,3}}^{u}\} = 0.4013$, 
the SNR order changes and $R_{1,3}$ is also reduced.
This issue did not arise in $c_2$ as there was no finite jammer power upper bound $P_{j_{1,2}}^{u}$. \hfill\textbf{$\Box$}

After decision of $\pi_{m,n}$ and $\pi_{j_n}$, combinatorial aspect of the problem is resolved if $P_{j_n}$ remains within bounds. Next we consider the joint source and jammer power allocation.

%
\begin{table}[!htb]
\caption{Users' SNRs and secure rate $R_{1,3}$ versus $P_{j_3}$}\label{carrier_3}
\centering
{\scriptsize
\begin{tabular}{|L{1cm} |L{1cm}| L{1cm}| L{1cm}| L{1cm}| L{1cm}| L{1cm}|}
\hline
{$P_{j_3}\to $} &    {$0.0$} & {$0.1$} & {$0.4$} & {$0.5$} & {$0.7$} \\
\hline
{$SNR_{1}$} & {0.4514} & {0.2086} & {0.0798} & {0.0662} & {0.0493} \\
\hline
{$SNR_{2}$} & {0.4301} & {0.1602} & {0.0556} & {0.0456} & {0.0336} \\
\hline
{$SNR_{3}$} & {0.0624} & {0.0605} & {0.0554} & {0.0539} & {0.0512} \\
\hline
{$R_{1,3}$} & {0.0328} & {0.1020} & {0.0616} & {0.0315} & {0.0000} \\
\hline
\end{tabular}
}
\end{table}

\subsection{Joint optimization of source and jammer power }\label{subsec_ri_jpa}
After the decision on utilization of jammer power based on Proposition \ref{prop_sri_jammer_rate_imp_constraints}, all subcarriers can be categorized in two sets: $\{ \mathcal{J}_0 \}$ - the ones that do not use jammer power ($\pi_{j_n}=0$), and $\{\mathcal{J}_1 \}$ - the others that use jammer power ($\pi_{j_n}=1$). 
Let $\gamma_{m,n}$ and $\gamma_{m,n}'$ respectively  denote the SNRs of the user $m$, without jammer power and with jammer power over subcarrier $n$, i.e., 
\begin{align}\label{eqn_jra_sja_snr_definition}
 \gamma_{m,n} = \frac{P_{s_n} |h_{m,n}|^2}{\sigma^2}; \quad  \gamma_{m,n}' = \frac{P_{s_n} |h_{m,n}|^2}{\sigma^2 + P_{j_n} |g_{m,n}|^2} .
\end{align}
The joint power allocation problem can be stated as
\begin{align}\label{opt_prob_jra_SPA_src_jamm_power}
& \underset{P_{s_x}, P_{s_y}, P_{j_y}}{\text{maximize}}
\left\{ \sum_{x \in \mathcal{J}_0} w_x'\left[ \log_2 ( 1+\gamma_{m,x} ) - \log_2 ( 1+\gamma_{e,x} )  \right] \right. \nonumber \\ 
& \left. \qquad + \sum_{y \in \mathcal{J}_1} w_y' \left[ \log_2 ( 1+ \gamma_{m,y}') - \log_2 ( 1+ \gamma_{e,y}' )  \right] \right \} \nonumber \\
& \text{subject to} \nonumber \\ 
& C_{2,1}: \sum_{x \in \mathcal{J}_0} P_{s_x} + \sum_{y \in \mathcal{J}_1} P_{s_y} \leq P_S,  \quad C_{2,2}: \sum_{y \in \mathcal{J}_1} P_{j_y} \leq P_J,\nonumber \\
& C_{2,3}: P_{j_y} > P_{j_{y}}^l \text{ } \forall \text{ } y \in \mathcal{J}_1, \quad C_{2,4}: P_{j_y} < P_{j_{y}}^u \text{ } \forall \text{ } y \in \mathcal{J}_1,\nonumber \\
& C_{2,5}: P_{s_x}\ge 0, P_{s_y}\ge 0 \text{ } \forall \text{ } x\in \mathcal{J}_0,y\in \mathcal{J}_1 
\end{align}
where $w_n'$ is the weight of the main user $m$ over the $n$th subcarrier.
$C_{2,1}$ and $C_{2,2}$ are the power budget constraints. 
Constraints $C_{2,3}$ and $C_{2,4}$ are imposed to tackle the SNR re-ordering challenge (cf. Section \ref{subsubsec_ri_bounds}), while the boundary constraints for source power are captured in $C_{2,5}$. 

The source power has to be shared among the subcarriers' set $\{ \mathcal{J}_0 \}$ and $\{ \mathcal{J}_1 \}$, while the jammer power has to be allocated on the subcarriers of set $\{ \mathcal{J}_1 \}$ only.
We observe that source power is a coupling variable between the power allocation problem over the complementary sets $\{ \mathcal{J}_0 \}$ and $\{ \mathcal{J}_1 \}$, and $C_{2,1}$ is the corresponding complicating  constraint.
Since the secure rate over a subcarrier is an increasing function of source power, the problem has to be solved at full source power budget $P_S$.
The optimization problem (\ref{opt_prob_jra_SPA_src_jamm_power}) can be solved using 
\emph{primal decomposition} (PD) procedure by dividing it into one master problem (outer loop) and two subproblems (inner loop) \cite{sboyd_decomposition}.
The first subproblem is source power allocation over set $\{ \mathcal{J}_0 \}$, and second is joint source and jammer power allocation over set $\{ \mathcal{J}_1 \}$.
The subproblem-1 can be stated as: \\
\textbf{\textit{Subproblem-1}}
\begin{align}\label{opt_prob_jra_src_power_p1}
& \underset{P_{s_x}}{\text{maximize}}
\sum_{x \in \mathcal{J}_0} w_x' \left[ \log_2  ( 1+\gamma_{m,x} ) - \log_2 ( 1+\gamma_{e,x} )  \right] \nonumber \\
& \text{subject to} \nonumber \\ 
& C_{3,1}: \sum_{x \in \mathcal{J}_0} P_{s_x} \leq t, \quad C_{3,2}: P_{s_x}\ge 0 \text{ } \forall \text{ } x \text{ } \in  \mathcal{J}_0
\end{align}
where $t$, representing the coupling variable, is the source power budget allotted to subproblem-1.
The source power budget to subproblem-2 is $P_S-t$, as described in the following:\\
\textbf{\textit{Subproblem-2}}
\begin{align}\label{opt_prob_jra_SPA_src_jamm_power_p2}
& \underset{P_{s_y}, P_{j_y}}{\text{maximize}}
\sum_{y \in \mathcal{J}_1} w_y' \left[ \log_2 ( 1+ \gamma_{m,y}') - \log_2 ( 1+ \gamma_{e,y}' ) \right] \nonumber\\
& \text{subject to} \nonumber \\ 
& C_{4,1}: \sum_{y \in \mathcal{J}_1} P_{s_y} \leq P_S - t, \quad C_{4,2}: \sum_{y \in \mathcal{J}_1} P_{j_y} \leq P_J, \nonumber \\
& C_{4,3}: P_{j_y} > P_{j_{y}}^l \text{ } \forall \text{ } y \in \mathcal{J}_1, \quad   C_{4,4}: P_{j_y} < P_{j_{y}}^u \text{ } \forall \text{ }y \in \mathcal{J}_1, \nonumber \\
& C_{4,5}: P_{s_y}\ge 0 \text{ } \forall \text{ } y \in \mathcal{J}_1.
\end{align}

The two subproblems are solved independently for a fixed $t$. Once the solutions of both the subproblems are obtained, the master problem is solved using subgradient method by updating the coupling variable $t$ as $t := t - \xi(\lambda_2-\lambda_1)$, where $\xi$ is an appropriate step size \cite{sboyd_subgradient} and $\lambda_1$, $\lambda_2$ are the Lagrange multipliers \cite{sboyd_decomposition} corresponding to the source power constraints in subproblems (\ref{opt_prob_jra_src_power_p1}) and (\ref{opt_prob_jra_SPA_src_jamm_power_p2}), respectively. This procedure of updating $t$ is repeated until either the sum secure rate saturates or the iteration count exceeds a threshold.

\subsubsection{Solution of subproblem-1} 
\label{subsubsec_ri_alloc_ps}
Since each subcarrier is allocated to the best gain user (cf. Section \ref{subsec_subc_alloc_src}), $|h_{m,x}| > |h_{e,x}|$ $\forall$ $x$ $\in$ $\{ \mathcal{J}_0 \}$. Thus, the objective function in (\ref{opt_prob_jra_src_power_p1}) is a concave function of $P_{s_{x}}$. The optimal $P_{s_{x}}^*$ obtained after solving the Karush-Kuhn-Tucker (KKT) conditions is given as
\begin{align}\label{eqn_woj_opt_src_power_allocation}
P_{s_{x}}^* = f \left( \frac{\sigma^2}{|h_{m,x}|^2},\frac{\sigma^2}{|h_{e,x}|^2},\frac{4w_x'}{\lambda_1 \ln 2} \right)
\end{align}
where $f$ is defined as follows: 
\begin{align}\label{function_definition}
f(\eta, \nu, \kappa) = - \frac{1}{2} \left[ \left(\nu+\eta\right) - \sqrt{\left(\nu-\eta\right)^2 + \kappa\left(\nu-\eta\right)  }\right]^+
\end{align}
$\lambda_1$, associated with $C_{3,1}$, is found such that $\sum_{x \in \mathcal{J}_0} P_{s_x} = t$.

\subsubsection{Solution of subproblem-2}\label{subsubsec_ri_alloc_ps_pj}
In order to solve the subproblem-2, we observe that the secure rate over the subcarriers of set $\{\mathcal{J}_1\}$ is a concave function of source power $P_{s_y}$ for a fixed jammer power $P_{j_y}$, and also a quasi-concave function of $P_{j_y}$ for a fixed $P_{s_y}$ as shown in Proposition \ref{prop_sri_jammer_rate_imp_constraints}. This motivates us to use the method of \emph{alternating optimization} (AO) \cite{sboyd_alternating} for joint power optimization, which alternates between source and jammer power optimizations. 
The AO method starts with $P_{s_y} = \frac{P_S-t}{|\mathcal{J}_1|}$, i.e., equal power over all the subcarriers of set $\{ \mathcal{J}_1 \}$, and for the known $P_{s_y}$ optimal $P_{j_y}$ allocation is done. For a fixed source power, the jammer power allocation problem can be stated as:
\begin{align}\label{opt_prob_jra_SPA_src_jamm_power_p2_eq_ps}
& \underset{P_{j_y}}{\text{maximize}}
\sum_{y \in \mathcal{J}_1} w_y' \left[ \log_2 ( 1+ \gamma_{m,y}') - \log_2 ( 1+ \gamma_{e,y}' )  \right] \nonumber \\
& \text{subject to} \nonumber \\ 
& C_{5,1}: \sum_{y \in \mathcal{J}_1} P_{j_y} \leq P_J, \qquad C_{5,2}: P_{j_y} > P_{j_{y}}^l \text{ } \forall \text{ } y \in \mathcal{J}_1, \nonumber  \\
& C_{5,3}: P_{j_y} < P_{j_{y}}^u \text{ } \forall \text{ } y \in \mathcal{J}_1.
\end{align}
Since the secure rate over a subcarrier $y \in \{ \mathcal{J}_1 \}$ achieves a maxima at a unique jammer power $P_{j_y}^{o}$ (cf. Proposition \ref{prop_sri_jammer_rate_imp_constraints}). So, first we evaluate optimal $P_{j_y}^\star \forall y$ for achieving maximum secure rate, respecting the jammer power bounds as follows:
\begin{align}\label{opt_jammer_power_in_bounds}
P_{j_y}^\star=
\begin{cases}
P_{j_{y}}^l + \delta, & \text{if $P_{j_y}^{o} < P_{j_{y}}^l$} \\
P_{j_{y}}^u - \delta, & \text{if $P_{j_y}^{o} > P_{j_{y}}^u$} \\
P_{j_y}^{o}, & \text{otherwise}
\end{cases} 
\end{align}
where $\delta$ is a small positive number. Note that, in secure rate improvement $P_{j_{y}}^l=0$ and $P_{j_y}^{o}>0$, so first case does not arise.
If the condition $\sum_{y \in \mathcal{J}_1} P_{j_y}^\star \leq P_J$ is satisfied, then optimum jammer power is allocated over all subcarriers, i.e., $P_{j_y} = P_{j_y}^\star \forall y \in \{ \mathcal{J}_1 \}$.
Otherwise the problem (\ref{opt_prob_jra_SPA_src_jamm_power_p2_eq_ps}) is solved under jammer power budget constraint.
The partial Lagrangian of the problem (\ref{opt_prob_jra_SPA_src_jamm_power_p2_eq_ps}) is given in (\ref{eqn_jra_ao_p2_lagrange_function}).
\begin{figure*}[!t]
\begin{align}\label{eqn_jra_ao_p2_lagrange_function}
L(P_j, \mu ) = \sum_{y \in \mathcal{J}_1} w_y' \left[ \log_2 \left( 1+\frac{P_{s_y}|h_{m,y}|^2}{\sigma^2 + P_{j_y} |g_{m,y}|^2} \right) 
 - \log_2 \left( 1+\frac{P_{s_y} |h_{e,y}|^2}{\sigma^2 + P_{j_y} |g_{e,y}|^2} \right)  \right] 
+ \mu \left( P_J - \sum_{y \in \mathcal{J}_1} P_{j_y}  \right) 
\end{align}
\hrulefill
\end{figure*}
After setting first-order derivative of the Lagrangian in (\ref{eqn_jra_ao_p2_lagrange_function}) equal to zero, we obtain a fourth-order nonlinear equation in $P_{j_y}$ having following form
\begin{align}\label{eqn_jra_ao_p2_non_linear_equation}
& \left( \frac{\mu \ln 2}{w'_y} \right)  \left( a_yP_{j_y}^4 + b_y P_{j_y}^3 + c_y P_{j_y}^2 + d_y P_{j_y} + e_y \right) \nonumber \\
& \qquad \qquad = P_{s_y} \left( c_y' P_{j_y}^2 + d_y' P_{j_y} + e_y' \right). 
\end{align}

The coefficients of the equation (\ref{eqn_jra_ao_p2_non_linear_equation}) are tabulated in Table \ref{non_linear_equation_coefficients}. In the domain $[0,P_{j_y}^\star)$, the secure rate is a concave increasing function (cf. Proposition \ref{prop_sri_jammer_rate_imp_constraints}) and since  $\sum_{y \in \mathcal{J}_1} P_{j_y}^\star > P_J$, depending on $\mu$, there exists a single positive real root $P_{j_y}^{r} ( < P_{j_y}^{\star})$ of (\ref{eqn_jra_ao_p2_non_linear_equation}). After obtaining $P_{j_y}^{r}$ for a fixed value of $\mu$, $\mu$ is updated using the subgradient method \cite{sboyd_subgradient}. 
For constrained jammer power budget, optimal jammer power $P_{j_y}^{\diamond}$ is obtained after constraining $P_{j_y}^{r}$ between jammer power bounds as in (\ref{opt_jammer_power_in_bounds}).
Thus, $P_{j_{y}}$ allocation for fixed $P_{s_y}$ can be written as: 
\begin{equation}\label{jammer_power_allocation_subpoblem_2}
P_{j_{y}}=
\begin{cases}
P_{j_y}^{\star}, & \text{if } \sum_{y \in \mathcal{J}_1} P_{j_y}^{\star} \leq P_J\\
P_{j_y}^{\diamond}, & \text{otherwise.}
\end{cases} 
\end{equation}
For a known $P_{j_{y}}$, the AO method now obtains $P_{s_{y}}$ based on (\ref{eqn_woj_opt_src_power_allocation}) with the source power's coefficient being $\frac{|h_{m,y}|^2}{\sigma^2 + P_{j_y}|g_{m,y}|^2}$, instead of $\frac{|h_{m,y}|^2}{\sigma^2}$. This completes one iteration of the AO method. This procedure of optimizing source and jammer powers alternatively continues until either the secure rate saturates or the iteration count exceeds a threshold. A PD and AO based Joint Power Allocation (JPA) scheme for weighted sum secure rate maximization is presented in Algorithm \ref{global_optimization_algo_rate_imp}. 
\begin{remark}
During power optimization in AO, if the jammer power over some of the subcarriers in $\{ \mathcal{J}_1 \}$ is zero, then these subcarriers are taken out of $\{ \mathcal{J}_1 \}$ and added back into $ \{ \mathcal{J}_0 \}$ for source power allocation in the next iteration of PD procedure.
\end{remark}

\begin{table*}[!tbh]
\renewcommand{\arraystretch}{1.}
\caption{Coefficients of the non linear equation (\ref{eqn_jra_ao_p2_non_linear_equation})} \label{non_linear_equation_coefficients}
\centering
\begingroup\makeatletter\def\f@size{9}\check@mathfonts
\begin{tabular}{l}
\hline
${a_y}$ = $\left( |g_{e,y}|^2 |g_{m,y}|^2 \right)^2$; \text{ }
${b_y}$ = $|g_{e,y}|^2 |g_{m,y}|^2 \left\{ 2\sigma^2 \left( |g_{e,y}|^2 + |g_{m,y}|^2 \right) + P_{s_y} \left(|g_{m,y}|^2 |h_{e,y}|^2 + |g_{e,y}|^2 |h_{m,y}|^2 \right) \right\}$; \\
${c_y}$ = $ P_{s_y}^2 |g_{e,y}|^2 |g_{m,y}|^2 |h_{e,y}|^2 |h_{m,y}|^2 + P_{s_y} \sigma^2 \left\{|g_{m,y}|^2|h_{e,y}|^2 \left( |g_{m,y}|^2 + 2|g_{e,y}|^2 \right)  + |g_{e,y}|^2 |h_{m,y}|^2 \left( |g_{e,y}|^2 + 2|g_{m,y}|^2 \right)  \right\}$ \\
$+ \sigma^4 \left\{ \left( |g_{e,y}|^2 \right) ^2 + \left( |g_{m,y}|^2 \right) ^2 +4|g_{e,y}|^2|g_{m,y}|^2 \right\}$; \\
${d_y}$ = $\left( |g_{e,y}|^2 + |g_{m,y}|^2 \right) \left\{ 2\sigma^4 + P_{s_y}\sigma^2 \left( |h_{m,y}|^2 + |h_{e,y}|^2 \right)  + P_{s_y}^2|h_{m,y}|^2|h_{e,y}|^2 \right\}  + P_{s_y}\sigma^4 \left(|g_{m,y}|^2 |h_{e,y}|^2 + |g_{e,y}|^2 |h_{m,y}|^2 \right)$; \\
${e_y}$ = $\sigma^4 \left( \sigma^2+P_{s_y}|h_{e,y}|^2 \right) \left( \sigma^2+P_{s_y}|h_{m,y}|^2 \right)$; \text{ }
${c_y'}$ = $|g_{e,y}|^2 |g_{m,y}|^2 \left(|g_{m,y}|^2 |h_{e,y}|^2 - |g_{e,y}|^2 |h_{m,y}|^2 \right)$; \\
${d_y'}$ = $2|g_{e,y}|^2|g_{m,y}|^2\sigma^2 \left( |h_{e,y}|^2-|h_{m,y}|^2 \right)$; \text{ }
${e_y'}$ = $\sigma^4 \left( |g_{e,y}|^2|h_{e,y}|^2-|g_{m,y}|^2|h_{m,y}|^2 \right) + \sigma^2P_{s_y}|h_{e,y}|^2|h_{m,y}|^2 \left( |g_{e,y}|^2-|g_{m,y}|^2 \right)$;\\
\hline
\end{tabular}
\endgroup
\end{table*}
  
\subsubsection{Convergence of joint power allocation} 
The source power allocation is done through PD with $(\lambda_2-\lambda_1)$ as the subgradient. The secure rate is a concave increasing function of source power in subproblem-1 and for a fixed jammer power, in subproblem-2. Thus, $\lambda_1$ and $\lambda_2$ are positive and bounded. Hence, $(\lambda_2-\lambda_1)$ is bounded and the method converges to an $\epsilon$-optimal value in finite number of steps \cite{sboyd_subgradient}.

The authors in \cite{HMWANG_TIFS_2014} have shown with the help of Bolzano-Weirerstrass theorem and \cite{Grippo_ORL_2000} that, AO method converges for a problem which is concave in one set of variable and quasi-concave with unique maxima in another set. Hence, the joint power allocation based on PD and AO converges. 


\begin{algorithm}
{\small 
\caption{Proposed Sum Secure Rate Maximization}
\label{global_optimization_algo_rate_imp}
\begin{algorithmic}[1]
\Procedure{}{}
\State \noindent\textbf{Subcarrier Allocation:} 
\For {every subcarrier $n$}
\State find $k = \arg \underset{m \in \left\{ 1,M \right\} } \max |h_{m,n}|$ 
\EndFor
\State \noindent\textbf{Subcarrier set creation:}
\State {Init} $\mathcal{J}_0 = \phi$ and $\mathcal{J}_1 = \phi$
\For {every subcarrier $n$} 
\If{$(|g_{e,n}|^2>|g_{m,n}|^2)$}
\State compute $P_{s_{n}}^{th_i}$ and $P_{j_{n}}^{th_i}$
\State $\mathcal{J}_1 = \mathcal{J}_1 + n$ 
\Else
\State $\mathcal{J}_0 = \mathcal{J}_0 + n$ 
\EndIf
\EndFor
\State \noindent\textbf{Source and Jammer Power Allocation:}
\State \underline{primal decomposition begins}
\State {Init} $t$, $iter_o = 0$ 
\State compute $P_{s_x} \text{ } \forall \text{ }x \text{ } \in \text{ } \mathcal{J}_0$ for $P_S = t$ according to (\ref{eqn_woj_opt_src_power_allocation})
\State \underline{alternate optimization begins}
\State {Init} $ iter_i = 0; P_{s_y} = \frac{P_S-t}{| \mathcal{J}_1 |} \text{ } \forall \text{ }y \text{ } \in \text{ } \mathcal{J}_1$
\State compute $P_{s_{y}}^{th_i}, P_{j_{y}}^{th_i}$, $P_{j_{y}}^{\diamond}$ and $P_{j_{y}}^{\star} \text{ } \forall \text{ }y \text{ } \in \text{ } \mathcal{J}_1$ 
\State allocate $P_{j_y}$ according to (\ref{jammer_power_allocation_subpoblem_2})
\State allocate $P_{s_y}$ with $P_S = P_S-t$ according to (\ref{eqn_woj_opt_src_power_allocation})
\State $iter_i = iter_i + 1$
\State iterate till either rate improves or $iter_i \leq iter_i^{max}$
\State \underline{alternate optimization ends}
\State update $t$
\State $iter_o = iter_o + 1;$
\State iterate till either rate improves or $iter_o \leq iter_o^{max}$
\State \underline{primal decomposition ends}
\EndProcedure
\end{algorithmic}
}
\end{algorithm}

\subsection{Solution with reduced complexity for rate maximization}
\label{easy_sol_rate_max}
The complexity involved in joint source and jammer power optimization (discussed in Section \ref{complexity_analysis_rate_imp}) can be reduced to a large extent by allocating source and jammer powers sequentially instead \cite{ravikant_ICC_2015}.
Initially assuming the jammer to be absent, subcarrier allocation is done using (\ref{subcarrier_alloc_policy}) and the source power optimization is done using (\ref{eqn_woj_opt_src_power_allocation}).
After creating subcarrier set $\{ I_m \}$ of user $m$ (using (\ref{subcarrier_alloc_jammer})), over which the secure rate can be improved, the resource allocation problem is attended with jammer power allocation.
The optimization problem of $P_{j_{m,n}}$ (to aid user $m$ over subcarrier $n$) can be written as:
\begin{align}\label{subopt_jammer_power_allocation_problem}
& \underset{P_{j_{m,n}}}{\text{maximize}}
  \sum_{m=1}^M w_m \left[ \sum_{n \in I_m} \left\{ \log_2 \left( \frac{1+\frac{P_{s_n} |h_{m,n}|^2}{\sigma^2 + P_{j_{m,n}} |g_{m,n}|^2}} {1+\frac{P_{s_n} |h_{e,n}|^2}{\sigma^2 + P_{j_{m,n}} |g_{e,n}|^2}} \right) \right\} \right] \nonumber \\
& \text{subject to}  \nonumber \\
& C_{6,1}: \sum_{m=1}^M \sum_{n \in I_m}  P_{j_{m,n}} \leq P_J, \quad C_{6,2}: P_{j_{m,n}} < P_{j_{m,n}}^u \forall \text{ } n \in I_m  \nonumber \\
& C_{6,3}: P_{j_{m,n}} > P_{j_{m,n}}^l  \forall \text{ } n  \in  I_m
\end{align}
where $P_{j_{m,n}}^l$ and $P_{j_{m,n}}^u$ are the lower and upper bounds on $P_{j_{m,n}}$, respectively.
If $x>y$, the inequality $\frac{1+x}{1+y}<\frac{x}{y}$ can be used to convert the objective function into a concave function which upper bounds the achievable secure rate of (\ref{subopt_jammer_power_allocation_problem}).

\begin{figure*}[!t]
\begin{align}\label{subopt_lagrange_function}
L(P_j, \lambda) = 
& \sum_{m=1}^M w_m \left[ \sum_{n \in I_m} \left\{ \log_2 \left( \frac{|h_{m,n}|^2}{\sigma^2 + P_{j_{m,n}}|g_{m,n}|^2} \right)
 - \log_2 \left( \frac{|h_{e,n}|^2}{\sigma^2 + P_{j_{m,n}}|g_{e,n}|^2} \right) \right\} \right]
+ \lambda(P_J - \sum_{m=1}^M \sum_{n \in I_m}  P_{j_{m,n}})
\end{align}
\hrulefill
\end{figure*}

The partial Lagrangian of the simplified problem is given in (\ref{subopt_lagrange_function}). Setting the first derivative of the Lagrangian equal to zero, we obtain a quadratic equation in $P_{j_{m,n}}$ as:
\begin{align}\label{rate_imp_quadratic}
 P_{j_{m,n}}^2 |g_{m,n}|^2 |g_{e,n}|^2 +  P_{j_{m,n}} \sigma^2 \left( |g_{m,n}|^2 + |g_{e,n}|^2 \right) \nonumber \\
 + \sigma^4  - \sigma^2 \left( \frac{w_m}{\lambda \ln 2} \right) \left( |g_{e,n}|^2 - |g_{m,n}|^2 \right)  = 0.
\end{align}
The solution $P_{j_{m,n}}^*$ of (\ref{rate_imp_quadratic}) is obtained as 
\begin{align}\label{sub_opt_jammer_power}
P_{j_{m,n}}^* = f \left( \frac{\sigma^2}{|g_{e,n}|^2},\frac{\sigma^2}{|g_{m,n}|^2},\frac{4w_m}{\lambda \ln 2} \right)
\end{align}
where $f$ is defined in (\ref{function_definition}) and $\lambda$ is obtained such that $\sum_{m=1}^M \sum_{n \in I_m}  P_{j_{m,n}}^* = P_J$. This scheme works fine till $\sum_{m=1}^M \sum_{n \in I_m} P_{j_{m,n}}^u$ $>$ $P_J$. When $\sum_{m=1}^M \sum_{n \in I_m} P_{j_{m,n}}^u$ $\leq$ $P_J$, the scheme allocates $P_{j_{m,n}}$ $=$ $P_{j_{m,n}}^u - \delta$ which results in degraded performance. In order to solve this issue we allocate $P_{j_{m,n}}$ $=$ $(P_{j_{m,n}}^l+P_{j_{m,n}}^u)/2$. Thus, $P_{j_{m,n}}$ is allocated as:
\begin{align}\label{jammer_power_allocation}
P_{j_{m,n}}=
\begin{cases}
P_{j_{m,n}}^*, & \text{if } \sum_{m=1}^M\sum_{n \in I_m} P_{j_{m,n}}^u > P_J\\
\frac{P_{j_{m,n}}^l + P_{j_{m,n}}^u}{2}, & \text{otherwise.}
\end{cases} 
\end{align}
The suboptimal scheme 
actually  maximizes the upper bound of the jammer power allocation problem (\ref{subopt_jammer_power_allocation_problem}).
As anticipated the suboptimal solution referred as JPASO has degraded secure rate 
performance compared to JPA as shown in Section \ref{sec_results}. 

\subsection{Complexity analysis}\label{complexity_analysis_rate_imp}
Here we analyze complexity of the proposed algorithms for secure rate improvement.
The computation complexity of jammer power upper and lower bound for each subcarrier is $O(NM)$. 
The PD procedure first optimizes source power over $N_1(=|\mathcal{J}_0|)$ subcarriers that do not utilize jammer power, in $N_1\log(N_1)$  computations \cite{W_Yu_TCOM_2006}.
The AO procedure optimizes source as well as jammer power over the leftover $N_2=N-N_1$ subcarriers alternatively in $N_2 \log(N_2)$ computations.  
Let us denote that AO takes $I_{ao}$ iterations and PD takes $I_{pd}$ subgradient updates for convergence. Then, the total complexity of JPA can be obtained as $O(NM + I_{pd} N_1 \log(N_1) + I_{ao} I_{pd} N_2 \log(N_2))$.
The worst case computational complexity 
when $N_2 \simeq N$, is $O(I_{ao} I_{pd} N \log(N))$.
JPASO initially optimizes source power over all the subcarriers in $O(N \log(N))$ computations and then optimizes jammer power over $N_2$ subcarriers with complexity $O(N_2 \log(N_2))$.
Thus, the worst case complexity of JPASO is $O(N \log(N))$.

\section{Fair resource allocation}
\label{sec_max_min_fairness}
Allocating a subcarrier to its best gain user is an optimal sum rate maximization strategy in non-secure OFDMA systems \cite{Rhee2000}, \cite{Jang2003}. 
In contrast, secure OFDMA with untrusted users in the absence of jammer has in general no other option but to allocate a subcarrier to its best gain user.
This may cause some users to starve for resources and may lead to poor fairness performance.
The concept of generalized $\alpha$-fairness was introduced in \cite{Altman_alpha_2008}, where $\alpha=0$ corresponds to sum rate maximization, $\alpha=1$ corresponds to proportional fairness, and $\alpha \to \infty$ leads to max-min fairness. 
With appropriate selection of the user priority weights, the performance of the first two scenarios, i.e., with $\alpha=0$ and $\alpha=1$, can be captured by the sum rate maximization strategy discussed in Section \ref{sec_sum_rate_imp}. 
In this section we explicitly look at the inherent challenges associated with max-min fair resource allocation in secure OFDMA.

\subsection{Max-min fairness scheme for secure OFDMA}
In secure OFDMA the max-min fair resource allocation problem in presence of jammer can be stated as follows:
\begin{align}\label{max_min_optimization_problem}
& \underset{P_{s_n}, P_{j_n}, \pi_{m,n}, \pi_{j_n}} {\text{maximize}} \underset{m}\min \sum_{n=1}^N \pi_{m,n} R_{m,n} \nonumber \\
& \text{subject to}  
 \quad C_{1,1}, C_{1,2}, C_{1,3}, C_{1,4}, C_{1,5}, C_{1,6} \text{ as in } (\ref{opt_prob_jra_rate_max}). 
\end{align} 
In non-secure OFDMA, because of high computation complexity involved in finding the optimal solution, the authors in \cite{Rhee2000, Mohanram_CL_2005} proposed suboptimal solutions for the max-min problem.
In secure OFDMA with untrusted users and a jammer, complexity of the problem increases further due to the presence of $\max$ operator in secure rate definition (cf. (\ref{secure_rate_definition})) and the addition of two variables, namely $\pi_{j_n}$ and $P_{j_n}$. 
Also, the conventional max-min algorithm \cite{Rhee2000} cannot be directly used in secure OFDMA because the algorithm may get stuck in an infinite loop. 
In other words, if a least-rate user does not have maximum SNR over any of the leftover subcarriers, this user gets the algorithm into a deadlock state in the `while loop', as its rate cannot be improved.
In order to address the deadlock problem associated with the conventional max-min fairness in our considered jammer assisted secure OFDMA, we propose a modified max-min fairness, which uses a novel concept of subcarrier snatching discussed below.

\subsection{Subcarrier snatching and maximum rate achievability}
In non-secure OFDMA, a user is poor if its average channel gain is relatively low compared to other users. 
The source can help such users by either implementing max-min fairness strategy, or transmitting at higher power, or using both. 
In secure OFDMA a poor user is the one who, in-spite of having good channel gain on main channel, has a very strong eavesdropper causing its secure rate close to zero.
This user cannot be helped much by increasing the source transmission power  
because rate of the eavesdropper also increases with source power.
Thus, with randomly distributed users, source has very limited maneuverability to help such a poor user.
Below, we introduce the concept of subcarrier snatching with the help of jammer, which is utilized in the proposed max-min fair resource allocation (described in Section \ref{proposed_max_min_scheme}). 

Let us reconsider an OFDMA system consisting of four nodes: source, jammer, and users $m$ and $e$. Let us assume that over a subcarrier $n$ gain of $e$ is higher than that of $m$ i.e., $|h_{e,n}|>|h_{m,n}|$. Originally this subcarrier should be allocated to best gain user $e$, but jammer power can be utilized to snatch this subcarrier from user $e$ and allocate to user $m$. The following Proposition describes the conditions of subcarrier snatching and the existence of unique maxima with $P_{j_n}$. 

\begin{proposition}\label{prop_max_min_subcarrier_snatching}
In secure OFDMA, a user $m$ can snatch subcarrier $n$ from its best gain user $e$ 
when channel gains are such that $|g_{e,n}|^2 |h_{m,n}|^2 > |g_{m,n}|^2 |h_{e,n}|^2 $ and jammer power is above the following threshold:
\begin{align}\label{eqn_pjn_subcarrier_snatching}
P_{j_n} >P_{j_{n}}^{th_s} = \frac {\sigma^2 \left( |h_{e,n}|^2 - |h_{m,n}|^2 \right)} {|g_{e,n}|^2 |h_{m,n}|^2 - |g_{m,n}|^2 |h_{e,n}|^2}.
\end{align}
The secure rate achieved over a snatched subcarrier has a quasi-concave nature, achieving a unique maxima with  $P_{j_n}$.
\end{proposition}

\begin{IEEEproof}
See Appendix \ref{sec_appendix_proposition_carr_snatch}.
\end{IEEEproof}

Thus, under certain channel conditions, a subcarrier can be snatched from its best gain user after allocating sufficient jammer power. The secure rate over the snatched subcarrier is not a monotonically increasing function of jammer power due to the interplay between channel coefficients, and source and jammer powers. In fact, there exists an optimal jammer power achieving maximum secure rate over the snatched subcarrier. 
Due to SNR reordering, (\ref{eqn_jra_SPA_pjn_eve_bound}) may result only in upper bounds as  $|h_{e,n}|>|h_{k,n}|$, while (\ref{eqn_jra_SPA_pjn_main_eve_bound}) leads to lower bounds (cf. \eqref{eqn_pjn_subcarrier_snatching}).

\emph{
Example (continued):} 
Following the strategy of allocating subcarrier to its best gain user, $c_1$ and $c_3$ can be allocated to $u_1$, and $c_2$, $c_4$, and $c_5$ can be allocated to $u_3$.
Hence, $u_2$ is left without any subcarrier.
If we follow conventional max-min approach \cite{Rhee2000}, both $u_1$ and $u_3$ are assigned a subcarrier each in the `for loop', but the algorithm will get stuck in the `while loop' as the secure rate of $u_2$ cannot be improved further. 
If jammer power is utilized, $u_2$ can snatch subcarrier $c_4$ from $u_3$.
We can observe this from the variation of secure rate $R_{2,4}$ of $u_2$ over $c_4$ with $P_{j_4}$.
The secure rate is $R_{2,4}=0.0, 0.2186, 0.5646, 0.5652, 0.5649$ when $P_{j_4} = 0.1, 0.2, 0.9, 0.9587, 1.0$, respectively.
$P_{j_4}^{th_s} = 0.1138$ and $P_{j_4}^o = 0.9587$ (cf. Proposition \ref{prop_max_min_subcarrier_snatching}).
Thus, $R_{2,4}>0$ when $P_{j_4}>P_{j_4}^{th_s}$, and $R_{2,4}$ has a maxima at $P_{j_4}=P_{j_4}^o$. \textbf{$\Box$}
 
\subsection{Proposed modified max-min fairness scheme}\label{proposed_max_min_scheme}
Our proposed max-min fairness scheme uses the concept of subcarrier snatching and also removes the possibility of algorithm getting locked in `while loop'.
Assuming equal source and jammer power allocation in the initialization phase, the algorithm finds two sets of subcarriers for each user $m$: $\{B_m\}$ and $\{S_m\}$. $\{B_m\}$, hereafter referred as best subcarriers of user $m$, contains those subcarriers over which user $m$ has the best channel gain, and $\{S_m\}$ contains those subcarrier which the user $m$ can snatch from its best gain user with the help of jammer. The algorithm initially assumes all the users to be active users, and stores the threshold jammer powers required for snatching the subcarriers (cf. (\ref{eqn_pjn_subcarrier_snatching})). A brief sketch of the proposed max-min algorithm is outlined as follows:

(1) Allocate one best subcarrier to each user in the `for loop' and update users' rate.

(2) Repeat the following steps till there are active users and leftover subcarriers

\hspace{1em}(i) Pick the user with lowest achieved rate

\hspace{1em}(ii) Find a best subcarrier for the user

\hspace{1em}(iii) If found, allocate it to the user and update its rate

\hspace{1em}(iv) Else, check if the user can snatch a subcarrier

\hspace{1em}(v) If yes, allocate the one with minimum $P_{j_{n}}^{th_s}$ and 

\hspace{2.5em} update its rate

\hspace{1em}(vi) Else, remove the user from the active users' list.

The proposed algorithm provides additional opportunity to a least-rate user in the form of snatching a subcarrier in Step (iv). The Proposition \ref{prop_max_min_subcarrier_snatching} suggests that the jammer power should be above a certain threshold for initiating subcarrier snatching. Also there exists an optimal jammer power that achieves the maximum rate over the snatched subcarrier. Thus, the utilization of jammer power at Step (v) can be done in two different ways.  Either a fixed jammer power is reserved over each subcarrier called \emph{proactively fair allocation}, or jammer power is allocated based on user demand, called \emph{on-demand allocation}.
For these two methods we present two different max-min fair schemes in following sections.
The Step (vi), provides a graceful exit of the algorithm from a possible bottleneck. 
It suggests that if a poor user cannot be helped beyond a stage, the user is taken out of the allocation loop.
\subsubsection{Proactively fair jammer power allocation (PFA)}
In this scheme total jammer power is divided equally among all the subcarriers, i.e., 
at each subcarrier there exists a jammer power budget of $P_{j}^{eq} = \frac{P_J}{N}$ that can be utilized by any user for snatching.
By reserving the jammer power on each subcarrier, the scheme attempts to provide equal opportunity of subcarrier snatching to all the users even on the last subcarrier.
The proposed max-min resource allocation scheme based on PFA is presented in Algorithm \ref{suboptimal_algo_max_min_with_jammer_1}.
Note that, this scheme optimizes source and jammer power after every subcarrier allocation.
Per-user source and jammer power optimizations are done as described in Section \ref{subsec_ri_jpa} with the only difference that, now the sets $\{ \mathcal{J}_0 \}$ and $\{ \mathcal{J}_1 \}$ contain the subcarriers of one user.

\subsubsection{On-demand jammer power allocation (ODA)}\label{max_min_oda}
The condition of subcarrier snatching requires $P_{j_n}>P_{j_{n}}^{th_s}$ on subcarrier $n$ (cf. Proposition \ref{prop_max_min_subcarrier_snatching}).  Thus, by limiting the jammer power budget to $P_{j}^{eq}$, PFA reduces the possibilities of subcarrier snatching. The ODA scheme dynamically allocates jammer power based on user demands on a first come first serve basis. While optimizing jammer power through AO, algorithm allocates optimal jammer power $(P_{j_n}^\star)$ required to achieve maximum secure rate over each snatched subcarrier instead of finding $P_{j_n}^\diamond$ (solution of  (\ref{eqn_jra_ao_p2_non_linear_equation})). The proposed max-min scheme in this case is similar to the scheme described in Algorithm \ref{suboptimal_algo_max_min_with_jammer_1} except for the following changes: 
\begin{itemize}
\item The algorithm calculates optimal jammer power for each possibility of subcarrier snatching in the initialization. 
\item While optimizing the jammer power over the subcarriers for a user, the algorithm sums up all the allocated jammer powers of the user and the leftover jammer power as the prospective jammer power budget.
The algorithm then allocates optimal jammer power over all the snatched subcarriers and updates the leftover jammer power.
\end{itemize}

ODA allocates jammer power dynamically, while PFA has conservative nature, it would be interesting to observe how these schemes compete in terms of fairness. We will discuss more about their relative performance in Section \ref{sec_results}.

\algrenewcomment[1]{\(\triangleright\) #1}
\begin{algorithm}
{\small 
\caption{Proposed Max-min Fair Resource Allocation}
\label{suboptimal_algo_max_min_with_jammer_1}
\begin{algorithmic}[1]
\Procedure{}{}
\State \emph{Initialization}
\State $U =\{1,2,\cdots M\}$; $C = \{1,2,\cdots N\}$; 
\State $R_m = 0 \text{ } \forall \text{ } m \in U$; \Comment{user rates}
\State $B_m = \phi \text{ } \forall \text{ } m \in U$; \Comment{best subcarriers}
\State $Ab_m = \phi \text{ } \forall \text{ } m \in U$; \Comment{allocated best subcarriers }
\State $S_m = \phi \text{ } \forall \text{ } m \in U$; \Comment{possible snatch subcarriers }
\State $As_m = \phi \text{ } \forall \text{ } m \in U$; \Comment{allocated snatched subcarriers}
\State $P_{s_n}^{eq} = P_S/N$  $P_{j_n}^{eq} = P_J/N$;  \Comment{equal power allocation}
\For {every user $m$} 
\State find $B_m =  \{ n: | h_{m,n} | > \underset{e \in \{1,M\} \setminus m} \max (| h_{e,n} |) \}$
\State find all $n$ such that $|g_{e,n}|^2 |h_{m,n}|^2 > |g_{m,n}|^2 |h_{e,n}|^2$
\State \hspace{0.25in} where $| h_{e,n} | = \underset{i \in \{1 \cdots M\}} \max (|h_{i,n}|)$ \& $e \neq m$
\State $S_m = S_m + n$; 
\State compute $P_{j_{m,n}}^l$ and $P_{j_{m,n}}^u \text{ } \forall{n \in S_m}$
\State find $n$ = $\arg \underset{n \in B_m } \max (| h_{m,n} |/| h_{e,n} |)$ 
\State $Ab_m = Ab_m + n $; $C = C - n$; $B_m = B_m - n$; 
\State update $R_m$
\EndFor
\While{($(U \neq \phi)$ \&\& $(C \neq \phi)$)}
\State find $v$ = $\arg \underset{m \in \{1 \cdots M\}} \min R_m$ \Comment{minimum rate user}
\If{$B_v \neq \phi$} \Comment{best subcarrier}
\State find $n$ = $\arg \underset{n \in B_v}\max (|h_{v,n}|/|h_{e,n}|)$
\State $Ab_v = Ab_v + n $; $C = C - n$; $B_v = B_v - n$;
\State optimize $P_{s_n}$ and update $R_v$
\Else
\If{$S_v \neq \phi$} \Comment{snatched subcarrier}
\State find $n = \arg \underset{n \in S_v} \min P_{j_{n}}^{th_s}$
\If{$ P_{j_n}^{eq}>=P_{j_{n}}^{th_s}$}
\State $As_v = As_v + n $; $C = C - n$; $S_v = S_v - n$;
\State optimize $P_{s_n}$ and $P_{j_{n}}$, and update $R_v$
\Else
\State $U = U - v$ \Comment{remove the user}
\EndIf
\Else
\State $U = U - v$ \Comment{remove the user}
\EndIf
\EndIf
\EndWhile
\EndProcedure
\end{algorithmic}
}
\end{algorithm}

\subsection{Modified max-min fairness with reduced complexity}
\label{easy_sol_max_fair}
We now present the suboptimal versions for PFA and ODA. 
\subsubsection{Suboptimal PFA (PFASO)}
In order to obtain a less complex solution, 
instead of optimizing source and jammer power jointly after each subcarrier allocation, 
we follow the strategy of sequential power allocation \cite{ravikant_ICC_2015}, as described in Section \ref{easy_sol_rate_max}.
Source power over the best subcarriers of a user is optimized according to (\ref{eqn_woj_opt_src_power_allocation}).
Keeping the equal source power fixed over the snatched subcarriers, the jammer power  is optimized by solving a single user jammer power allocation problem on
similar lines as (\ref{subopt_jammer_power_allocation_problem}).
Assuming $\{As_m\}$ to be the set of allocated snatched subcarriers of user $m$,
the jammer power budget is given as $P_J|As_m|/N$.
Since a subcarrier can be snatched only when $P_J/N$ $>$ $P_{j_{m,n}}^l$, 
thus $\frac{P_J | As_m |}{N}$ $>$ $\sum_{n \in As_m} P_{j_{m,n}}^l$ and hence   (\ref{subopt_jammer_power_allocation_problem}) has a feasible solution.
Till $\sum_{n \in As_m} P_{j_{m,n}}^u > \frac{P_J | As_m |}{N}$, the optimal jammer power $P_{j_{m,n}}^*$ (cf. (\ref{sub_opt_jammer_power})) is utilized.
However, when there is enough jammer power, the assignment is $P_{j_{m,n}} = P_{j_{m,n}}^u -\delta$, because in   snatching scenario the secure rate is positive when $P_{j_{m,n}} > P_{j_{m,n}}^l$.
Combinedly, jammer power is allocated as:
\begin{equation}\label{jammer_power_allocation_maxmin}
P_{j_{m,n}}=
\begin{cases}
P_{j_{m,n}}^*, & \text{if } \sum_{n \in As_m} P_{j_{m,n}}^u > \frac{P_J | As_m |}{N}\\
P_{j_{m,n}}^u -\delta, & \text{otherwise}
\end{cases} 
\end{equation}
where $\delta$ is a small positive number. 
\subsubsection{Suboptimal ODA (ODASO)}
It is a very light weight scheme and does not attempt any power optimization.
With equal source power allocation, every time the algorithm attempts to snatch a subcarrier, it simply allocates the optimal jammer power $(P_{j_n}^\star)$ over the snatched subcarrier and updates 
leftover jammer power.
Once the jammer power is exhausted, the algorithm either allocates the best subcarrier to the least-rate user or takes it out from active user set.

\subsection{Complexity analysis}
While creating $\{S_m\}$, PFA considers for each user the possibility of subcarrier snatching from the best gain user. This involves calculation of jammer power thresholds for snatching, lower and upper jammer power bounds, and optimal jammer powers required for achieving maximum rate.
The total computation is $O(NM^2)$. 
Denote that a user got $N_1$ best subcarriers and $N_2$ snatched subcarriers before being taken out of the resource allocation loop. The joint power optimization procedure has the complexity of $O(I_{pd}  N_1 \log(N_1) + I_{ao} I_{pd}  N_2 \log(N_2))$, where $I_{ao}$ and $I_{pd}$ are respectively the number of AO and PD iterations.
The total computational on this user in all the previous iterations is of 
$O ( (N_2-1)I_{pd}  N_1 \log(N_1) + I_{ao} I_{pd} \sum_{i=2}^{N_2-1}i \log(i) + \sum_{i=2}^{N_1}i \log(i) )$ which can be upper bounded as $O( N_2 I_{pd} N_1 \log(N_1) + I_{ao} I_{pd} N_2^2 \log(N_2) + N_1^2 \log(N_1) )$ using   $\sum_{i=2}^{n}i \log(i) \leq (n-1)(n\log(n))<n^2 \log(n)$.
The worst case complexity of the algorithm when $N_2 \simeq N$ is $I_{ao} I_{pd} N^2 \log(N)$.

The computation complexity of ODA is on the same order as that of PFA. 
PFASO has worst case computation complexity $O(N^2 \log N)$, because it does not use PD and finishes source and jammer power optimization sequentially. 
ODASO does not perform any power optimization at all, and has the worst case computation complexity of $O(N^2M)$. 
The worst case computational complexities of PFA and ODA and their respective suboptimal versions are summarized in Table \ref{max_min_algorithm_comparison}.
\begin{table}[!htb]
\caption{Comparison of max-min algorithms}\label{max_min_algorithm_comparison}
\centering
{\scriptsize
\begin{tabular}{|L{0.9cm} |L{2cm}| L{1.5cm}| L{2.5cm}| }
\hline
&    {Power optimization} &{$P_j$ allocation} & {Complexity} \\
\hline
\hline
{PFA}  &  Joint &$P_{j_n}^\star$ or $P_{j_n}^\diamond$ & $O(I_{ao} I_{pd} N^2 \log(N))$ \\
\hline
{ODA}  &  Joint & $P_{j_n}^\star$ & $O(I_{ao} I_{pd} N^2 \log(N))$ \\
\hline
{PFASO}  &  Sequential & $P_{j_n}^*$ or $P_{j_{n}}^u$ &  $O(N^2 \log(N))$ \\
\hline
{ODASO}  & No & $P_{j_n}^\star$ & $O(N^2M)$ \\
\hline
\end{tabular}
}
\end{table} 

\section{Asymptotic analysis}\label{sec_asymp_anlysis}
In this section, we are interested to determine the best user, the strongest eavesdropper, and the maximum secure rate achievable when $\{P_{s_n}, P_{j_n}\} \to \infty$ on a subcarrier $n$. We note that, for a fixed main user and eavesdropper, a subcarrier can be used in without-jammer mode or with-jammer mode in either rate improvement or subcarrier snatching scenario. In without-jammer mode the secure rate is an increasing concave function of source power. In both the scenarios with jammer, secure rate is a concave increasing function of $P_{s_n}$ for a fixed $P_{j_n}$, and for a fixed $P_{s_n}$ there exists an optimal jammer power $P_{j_n}^\star$ which too is an increasing function of $P_{s_n}$. Hence, at $P_{j_n}^\star$ secure rate becomes an increasing function of $P_{s_n}$.  

When jammer is not active on a subcarrier, the user with the maximum channel gain $|h_{m,n}|$ is the main user. But when jammer is active this might not be true. To find the main user of the subcarrier when $\{P_{s_n}, P_{j_n}\} \to \infty$, we use the concept of jammer power bounds discussed in Section \ref{subsubsec_ri_bounds}. Specifically, a user $m$ will remain the main user if the condition $\gamma_{m,n}'>\gamma_{e,n}' \forall e \in \{ 1,\cdots, M \} \setminus m$ gives no upper bound on $P_{j_n}$. Denote the effective lower bound of these conditions  as $P_{j_n}^l$. We claim that only one user will satisfy this condition. Otherwise, let two users $m_1$ and $m_2$ be the contenders. Let the lower bounds on $P_{j_n}$ be $P_{j_n}^{l_{m_1}}$ and $P_{j_n}^{l_{m_2}}$. Without loss of generality let $P_{j_n}^{l_{m_1}} < P_{j_n}^{l_{m_2}}$. This contradicts the fact that user $m_1$ can be the main user at $P_{j_n} \to \infty$, because when $P_{j_n} > P_{j_n}^{l_{m_2}}$, $m_2$ becomes the main user.

Let the main user over a subcarrier be $m$. To identify the eavesdropper and the secure rate achievable at $\{P_{s_n}, P_{j_n}\} \to \infty$  we follow the following approach. We consider each of the leftover users as the possible eavesdroppers. Thus, there are $(M-1)$ $\{m, e\}$ user pairs, where $e \in \{ 1\cdots M \} \setminus m$. Over each of the pair, there exist four cases:

(1) If $|h_{m,n}|>|h_{e,n}|$ and $|g_{m,n}|>|g_{e,n}|$, rate improvement condition is not satisfied.
An obvious decision is to set $P_{j_n} = 0$. In without-jammer case, secure rate is a concave increasing function of $P_{s_n}$ and tends to $\log_2\left( \frac{|h_{m,n}|^2}{|h_{e,n}|^2} \right)$ as $P_{s_n} \to \infty$.

(2) If $|h_{m,n}|>|h_{e,n}|$ and $|g_{m,n}|<|g_{e,n}|$, rate improvement condition is satisfied.
$P_{j_{n}}^{th_{i}}$ is an increasing function of $P_{s_n}$ (cf. Proposition \ref{prop_sri_jammer_rate_imp_constraints}).
Hence, as $P_{s_n} \to \infty$, $P_{j_{n}}^{th_{i}} \to \infty$.
Putting $P_{j_n} = P_{j_n}^\star$, the optimal secure rate is $\log_2 \left( \frac{1 + \gamma_{m,n}'}{1 + \gamma_{e,n}'} \right)$, where 
\begin{align}\label{sec_rate_k_e_pair}
\frac{1 + \gamma_{m,n}'}{1 + \gamma_{e,n}'} &= \frac{1+\frac{P_{s_n} |h_{m,n}|^2}{\sigma^2 + P_{j_n}^\star |g_{m,n}|^2}} {1+\frac{P_{s_n} |h_{e,n}|^2}{\sigma^2 + P_{j_n}^\star |g_{e,n}|^2}} 
\stackrel{a}{\approx} 
\frac{\frac{P_{s_n} |h_{m,n}|^2}{\sigma^2 + P_{j_n}^\star |g_{m,n}|^2}} {\frac{P_{s_n} |h_{e,n}|^2}{\sigma^2 + P_{j_n}^\star |g_{e,n}|^2}} \nonumber \\
& = 
\frac{|h_{m,n}|^2}{|h_{e,n}|^2} \frac{\frac{\sigma^2}{P_{j_n}^\star} + |g_{e,n}|^2}{\frac{\sigma^2}{P_{j_n}^\star} + |g_{m,n}|^2} 
\stackrel{b}{\to}
\frac{|g_{e,n}|^2 |h_{m,n}|^2 }{|g_{m,n}|^2 |h_{e,n}|^2 }.
\end{align}
The approximation (a), and (b) follows from $P_{j_n}^\star$ being an increasing function of $P_{s_n}$, which increases as $\sqrt{P_{s_n}}$. Thus, with $P_{s_n} \to \infty$ and $P_{j_n} \to \infty$ we have $P_{j_n}^\star \to \infty$. 
 
(3) If $|h_{m,n}|<|h_{e,n}|$ and $|g_{e,n}||h_{m,n}|<|g_{m,n}||h_{e,n}|$, subcarrier snatching condition is not fulfilled; secure rate $=0$.

(4) If $|h_{m,n}|<|h_{e,n}|$ and $|g_{e,n}||h_{m,n}|>|g_{m,n}||h_{e,n}|$, subcarrier snatching condition is satisfied. The secure rate in this case is the same as in case (2). 

Once the secure rate for each user pair $\{m, e\}$ is obtained, 
the strongest eavesdropper for user $m$ at $\{P_{s_n}, P_{j_n}\} \to \infty$ is the one that causes the minimum secure rate. 
Thus, with jammer, for each subcarrier we have the main user, its eavesdropper, and the corresponding secure rate achieved.

When the jammer is inactive, the best channel gain user is the main user and the next best gain user is the corresponding eavesdropper. The secure rate in this case is: $\log_2\left( \frac{|h_{m,n}|^2}{|h_{e,n}|^2} \right)$.

\subsection{Asymptotic bound in sum rate maximization scenario}\label{upper_bound_rate_max}
Motivated by the asymptotic behavior of secure rate and identification of the main user and eavesdropper as  $\{P_{s_n}, P_{j_n}\} \to \infty$, we obtain an upper bound to the maximum secure rate of the system.
To do so, we consider a subcarrier without-jammer as well as with-jammer, and  choose the mode that offers the maximum secure rate upper bound over that subcarrier.
The system upper bound is then found as the sum of the upper bounds in chosen modes over all the subcarriers. 

To find asymptotically optimal solution for rate maximization we note that 
the sum rate achievable as $\{P_S, P_J\} \to \infty$ is an upper bound to the maximum sum rate of the system. Hence, the decision of subcarrier allocation and jammer mode at $\{P_S, P_J\} \to \infty$ is optimal.
Keeping the decision at $\{P_S, P_J\} \to \infty$ same over every finite $P_S$ and assuming $P_J \to \infty$, we find the sum rate that becomes optimal as $P_S \to \infty$. Let the sum rate be $R_U(P_S)$.
As $\{P_S, P_J\} \to \infty$, relative gap between the rate achievable in the proposed algorithm and $R_U(P_S)$ at any $P_S$ will demonstrate effectiveness of the algorithm in reaching the optimal point.
We note that at low values of $P_S$, $\{P_S, P_J\} \to \infty$ may not maximize the sum rate; hence it may be below the rate achievable by proposed scheme. 
The motivation behind $R_U(P_S)$ is that, for a fixed subcarrier allocation and jammer mode decision, the achievable maximum rate with finite $P_S$  and $P_J \to \infty$ is an upper bound to that achievable with finite $P_S$ and $P_J$. 

\subsection{Asymptotic bound in the proposed max-min fairness}\label{upper_bound_max_min}
In this case, we cannot use the per-subcarrier secure rate upper bounds of the two jammer modes. The proposed max-min fair scheme works on a per-subcarrier basis and updates rate of the minimum-rate user after every subcarrier allocation. Hence, if we apply the proposed max-min scheme on the upper bounds, it may not lead to system upper bound. Also, since the decision of subcarrier allocation and jammer usage varies with $P_S$, the decision at $P_S \to \infty$ cannot be kept fixed.

Instead, motivated by the ODA scheme in Section \ref{max_min_oda}, we note that as $P_J \to \infty$ the fairness achievable by ODA will be an upper bound to the fairness achieved by any other bounded $P_J$ scheme. 
In ODA scheme with finite $P_J$, the jammer power gets exhausted very soon.
Hence, it allows fewer subcarrier snatching, which limits its fairness performance.
Considering ODA with $P_J \to \infty$, we allocate $P_{j_n} = P_{j_n}^\star$ and optimize source power over all the subcarriers of the minimum rate user.  
Thus ODA can now help the poor users to snatch subcarriers and improve their rates without any jammer power budget constraint.
The best fairness achievable by our proposed scheme is given by this achievable upper bound.

\section{Results and discussion}
\label{sec_results}
Performances of the proposed algorithms are evaluated through MATLAB simulations. 
We consider the downlink of an OFDMA system with $N=64$ subcarriers, which are assumed to experience  frequency-flat slow fading. 
Large scale fading is modeled using a path loss exponent equal to $3$, and small scale fading is modeled using i.i.d. Rayleigh faded random variables.
The source is located at the origin and all untrusted users with default $M = 8$ are randomly located inside a unity square in the first quadrant which symbolically maps to a sector in cellular communications system. 
Friendly jammer's location is varied to determine its optimum position. For illustration, all users are assumed to have equal weights i.e., $w_i = 1$, and the AWGN noise variance $\sigma^2=1$. Secure rate is measured in bits per OFDM symbol per subcarrier. 

\textit{Equal power allocation:} We compare the performance of our proposed schemes with equal power allocation (EPA) that utilizes equal source power on all the subcarriers. 
Respecting the jammer power bounds developed in Section \ref{subsubsec_ri_bounds}, equal jammer power is assigned over the selected set of subcarriers following the observations in Proposition \ref{prop_sri_jammer_rate_imp_constraints} and \ref{prop_max_min_subcarrier_snatching} for rate improvement and max-min fairness scenarios, respectively. 

\textit{Optimal Source Power Without Jammer (OSPWJ):} In sum rate maximization case, this scheme obtains the optimal source power allocation assuming the jammer to be absent. In max-min fair scenario, after subcarrier allocation to a user, source power over the subcarriers allocated to the user is optimized.

The conventional fairness measure as in \cite{jain1984}, that considers user's rate as the basis for fairness evaluation, is not meaningful in our context because over a large number of iterations with random user locations, the users' rates tends to be similar, causing the fairness index to be close to unity. 
Instead, in order to compare the true capability of a fairness algorithm in bridging the imbalance in users' rate, we measure fairness as the relative gap between maximum and minimum secure rate allocated to the users by a competing algorithm. In this case, in  each iteration identity of the users are ignored and the secure rates allocated by the algorithm are sorted in ascending order and then fairness is obtained from this sorted rate. 

\subsection{Rate maximization schemes}
We first discuss the effect of jammer location on the system performance.
Four possible jammer locations $(0.5,0.5)$, $(0,1)$, $(1,0)$, $(1,1)$ are considered.
Fig. \ref{fig:rate_imp_jammer_location} presents the rate maximization algorithm performance versus source power. 
With varying locations of jammer, the algorithm performs differently, because of the varying degree of impact the jammer has on the users. 
Performance at location $(0,1)$ and $(1,0)$ are similar, because these are symmetric locations for the square geometry considered.
Location $(1,1)$ performs the poorest as the jammer is too far to have significant effect on users' rates. The central location $(0.5, 0.5)$ performs the best, as the jammer is able to affect all the users and contribute significantly in improving their rates. Motivated by this, we consider the jammer to be located at the center in our subsequent simulations.

\begin{figure}[!htb]
\centering
\epsfig{file=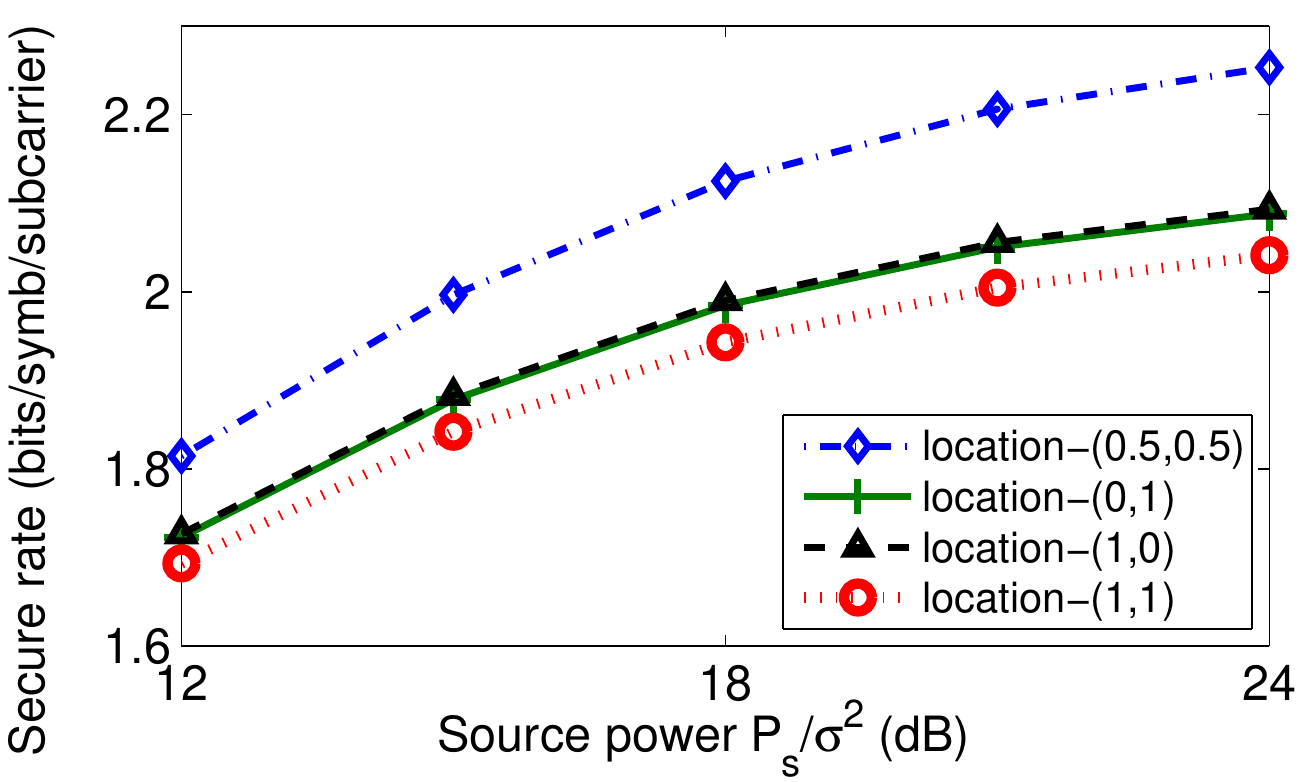, height=1.6in}
\caption{Secure rate versus source power at various jammer locations with jammer power $P_J/\sigma^2=0$ dB.}
\label{fig:rate_imp_jammer_location}
\end{figure}

Fig. \ref{fig:rate_imp_ps_variation} shows the secure rate and fairness performance of proposed JPA scheme with respect to source power at two different values of jammer power, $P_{J1}/\sigma^2=0$ dB and $P_{J2}/\sigma^2=6$ dB. The performance of JPA is compared with the suboptimal version JPASO and EPA. The rate achieved with OSPWJ is also plotted to emphasize secure rate improvement with friendly jammer. The performance of asymptotically optimal solution ($P_J \to \infty$) is plotted as `Asymp opt' to indicate closeness of proposed algorithms to the  optimal solution. Sum secure rate and the corresponding fairness upper bounds as $\{P_S, P_J\} \to \infty$ are also indicated in text boxes in the respective figures.

Fig. \ref{fig:rate_imp_ps_variation}(a) indicates that JPASO performs better compared to EPA because JPASO performs power optimization while EPA does not. JPASO has a marginal performance penalty with respect to JPA. Also, the gap between JPA and OSPWJ initially keeps increasing and then saturates because of diminishing returns at higher source power. 
It may be further noted that, `Asymp opt' has poor performance at lower source power budget, because the decision of subcarrier allocation and jammer mode at $\{P_S, P_J\} \to \infty$ may not be optimal at lower value of $P_S$. At finite $P_S$, the possibility of utilizing jammer power under jammer power bounds over larger number of subcarriers is more compared to that at $P_S \to \infty$. But, as $P_S$ increases the decision of subcarrier allocation and jammer mode becomes optimal and the `Asymp opt' tries to catch the upper bound.
Since secure rate performance of the proposed schemes improve with $P_J$, hence at higher $P_S$ and $P_J$ the performance of JPA tends to that of `Asymp opt'.

Fig. \ref{fig:rate_imp_ps_variation}(b) shows the associated fairness performance.
Because of equal distribution of resources in EPA, its fairness performance is better at lower source power compared to both JPA and JPASO. 
With increasing source power, the jammer is able to affect more number of users because the percentage of subcarriers over which jammer can help keep increasing (cf. Proposition \ref{prop_sri_jammer_rate_imp_constraints}).
This results in improved secure rate as well as fairness performance of JPA compared to EPA.
\begin{figure}[!htb]
\centering
\epsfig{file=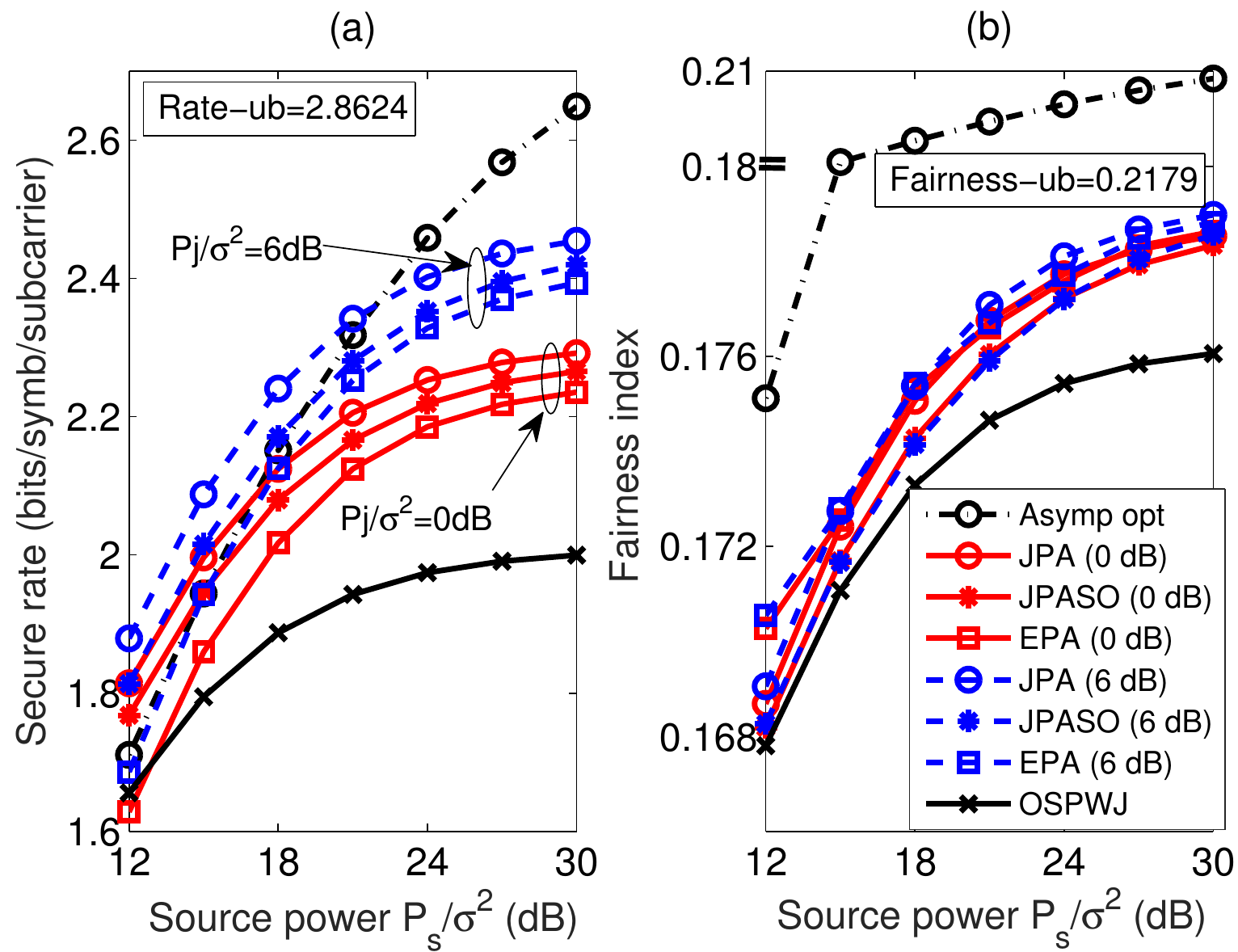, width=3.2in} 
\caption{Secure rate and fairness performance versus source power at $P_{J1}/\sigma^2 = 0$ dB and $P_{J2}/\sigma^2=6$ dB. `Rate-ub': rate upper bound; `Fairness-ub': fairness upper bound.}
\label{fig:rate_imp_ps_variation}
\end{figure}

Performance of the secure rate improvement algorithms with respect to jammer power is presented in Fig. \ref{fig:rate_imp_pj_variation}.
As observed in Fig. \ref{fig:rate_imp_pj_variation}(a) the secure rate of JPA initially increases with jammer power and then saturates, as the algorithm start allocating optimal jammer power $(P_{j_n}^\star)$ achieving maximum secure rate over the selected set of subcarriers (cf. Proposition \ref{prop_sri_jammer_rate_imp_constraints}).
Due to sequential power allocation, JPASO performance is inferior compared to JPA.
The secure rate saturates at a value lower than the peak value, because JPASO is oblivious to existence of optimal jammer power and utilizes more jammer power than required.
Since EPA uses equal jammer power, the secure rate initially increases, achieves a maximum and then reduces with increased jammer power.
Compared to EPA which utilizes equal jammer power blindly, the performance of JPASO is relatively better at higher jammer power, which can be attributed to improved jammer power allocation policy (\ref{jammer_power_allocation}). 
Because of the existence of finite upper bounds on jammer power (cf.  Section \ref{subsubsec_ri_bounds}), the rate achieved by EPA also saturates at higher jammer power but at a relatively lower value compared to JPASO.
The corresponding fairness performance plots for the various schemes are presented in Fig.  \ref{fig:rate_imp_pj_variation}(b).
At lower source power, EPA performance is better because of the utilization of equitable distribution of resources, while at higher source power JPA performs comparable to EPA at lower jammer power but outperforms at higher jammer power.
With increasing jammer power JPA fairness saturates while EPA fairness achieves a peak, reduces a bit and then saturates, showing similar trend as the corresponding secure rate plots.  
JPASO first completes source power optimization and then finishes jammer power optimization over the selected set of subcarriers, thereby increasing the secure rate imbalance which results in comparatively poor fairness performance.

\begin{figure}[!htb]
\centering
\epsfig{file=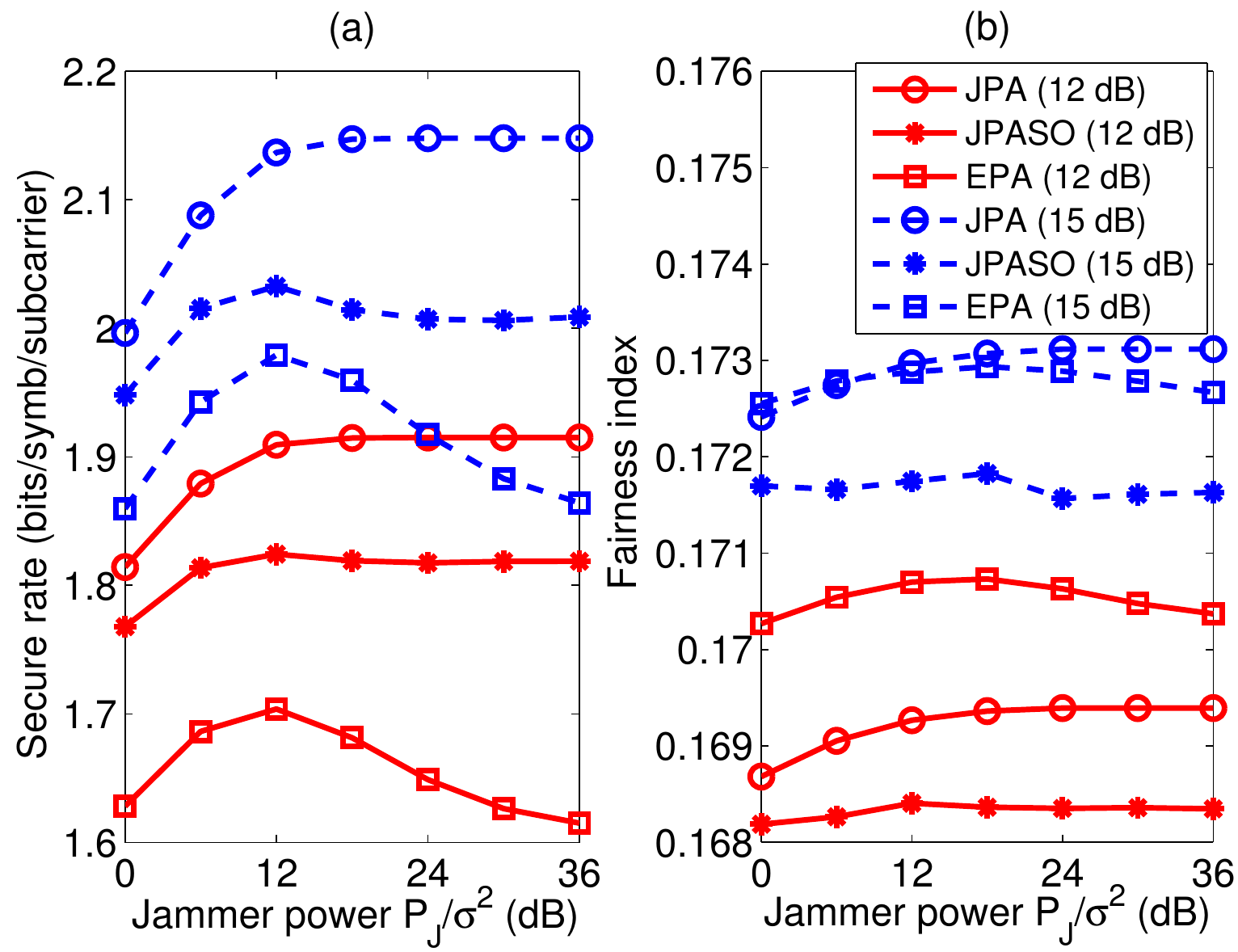,width=3.2in}
\caption{Secure rate and fairness performance versus jammer power at $P_{S1}/\sigma^2 = 12$ dB and $P_{S2}/\sigma^2=15$ dB.}
\label{fig:rate_imp_pj_variation}
\end{figure}

Fig. \ref{fig:rate_imp_user_var} presents the performance of the algorithms in secure rate improvement scenario with number of users $M$.
It may appear that, with increasing number of users the secure rate should reduce as the number of eavesdropper increases, but eventually 
the secure rate of all the algorithms improves with increasing number of users because of multi-user diversity. 

\begin{figure}[!htb]
\centering
\epsfig{file=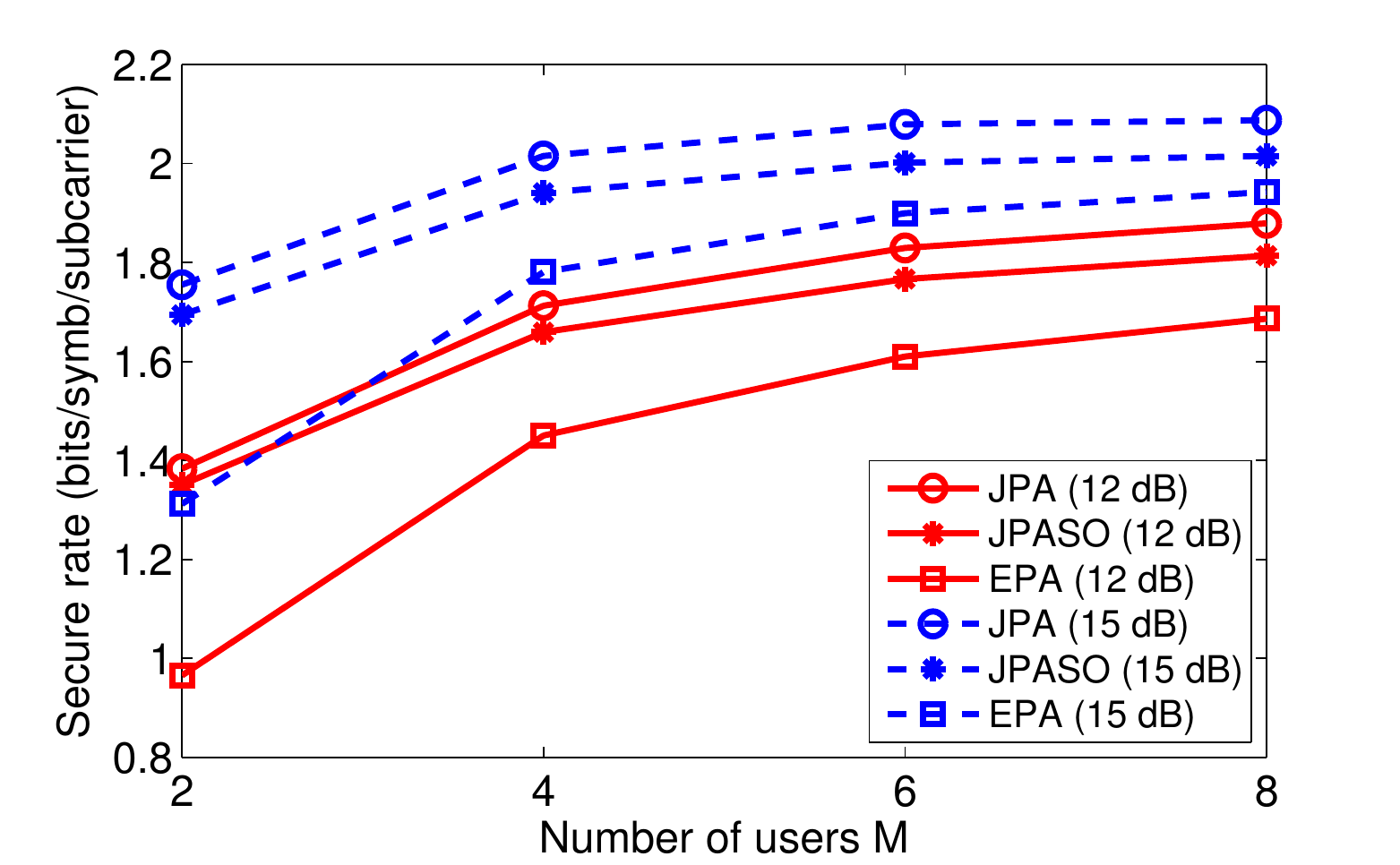,width=2.3in}
\caption{Secure rate versus number of users $M$ at $P_J/\sigma^2=6$ dB and $P_{S1}/\sigma^2 = 12$ dB and $P_{S2}/\sigma^2=15$ dB.}
\label{fig:rate_imp_user_var}
\end{figure}

\subsection{Max-min fair schemes}
Next we present the performance of the proposed max-min fairness algorithms in Figs.   \ref{fig:max_min_ps_variation} and \ref{fig:max_min_pj_variation}.
PFA and ODA and their corresponding suboptimal versions PFASO and ODASO are considered with EPA. 
The performance of asymptotically optimal scheme which act as an upper bound to our proposed max-min scheme has been plotted as `Asymp opt'.
Fig. \ref{fig:max_min_ps_variation} presents the fairness and secure rate performance of the algorithms versus source power.
ODA utilizing jammer power dynamically, allocates optimal jammer power on first come first serve basis.
This results in more subcarrier snatching possibilities and hence a higher fairness at low source power compared to PFA. Being a conservative algorithm, PFA is not able to help many users because of limited per subcarrier jammer power budget ($P_J/N$).
Since the optimal jammer power $(P_{j_n}^\star)$ for maximum secure rate over a snatched subcarrier is an increasing function of source power (cf. (\ref{rate_derivative_pjn_quad_coef1})-(\ref{rate_derivative_pjn_quad_coef3})), 
jammer resource gets exhausted very soon in ODA.
Because of this, ODA's capability to improve fairness reduces drastically with increasing source power, and correspondingly its fairness performance starts degrading. 
ODASO also faces the same issue of depleting all the jammer resource for a few initial users and later finding itself unable to help users, which results in early saturation of the fairness with source power.
In contrast PFA and EPA even if helping a limited number of users because of their conservative nature, keep performing well with increasing source power and outperform ODA at higher source power. 
While ODASO is able to help only a few initial users, EPA provides equal opportunity to all the users. As a result, the overall performance of EPA is better than ODASO. 
Though ODA at $P_J=18$ dB has similar performance as with $P_J=12$ dB, the plot gets closer to that of `Asymp opt' as $P_J$ is increased. `Asymp opt' does not face the problem of depleting jammer resources and keep on helping users as far as possible, which results in better system fairness compared to all other bounded $P_J$ schemes. The cost being paid for higher fairness is the reduced secure rate, as observed in Fig. \ref{fig:max_min_ps_variation}(b). 
The secure rate of OSPWJ is better than all other schemes. However, its fairness performance is the poorest because it either allocates the best subcarrier to a user or does not allocate at all. 
Because of the dynamic allocation, ODA allows more subcarrier snatching which results in poorer secure rate performance compared to PFA as observed in Fig. \ref{fig:max_min_ps_variation}(b).
Note that, with the help of jammer the overall system fairness can be improved significantly (compared to OSPWJ), however at the cost of a little poorer secure rate performance.

\begin{figure}[!htb]
\centering
\epsfig{file=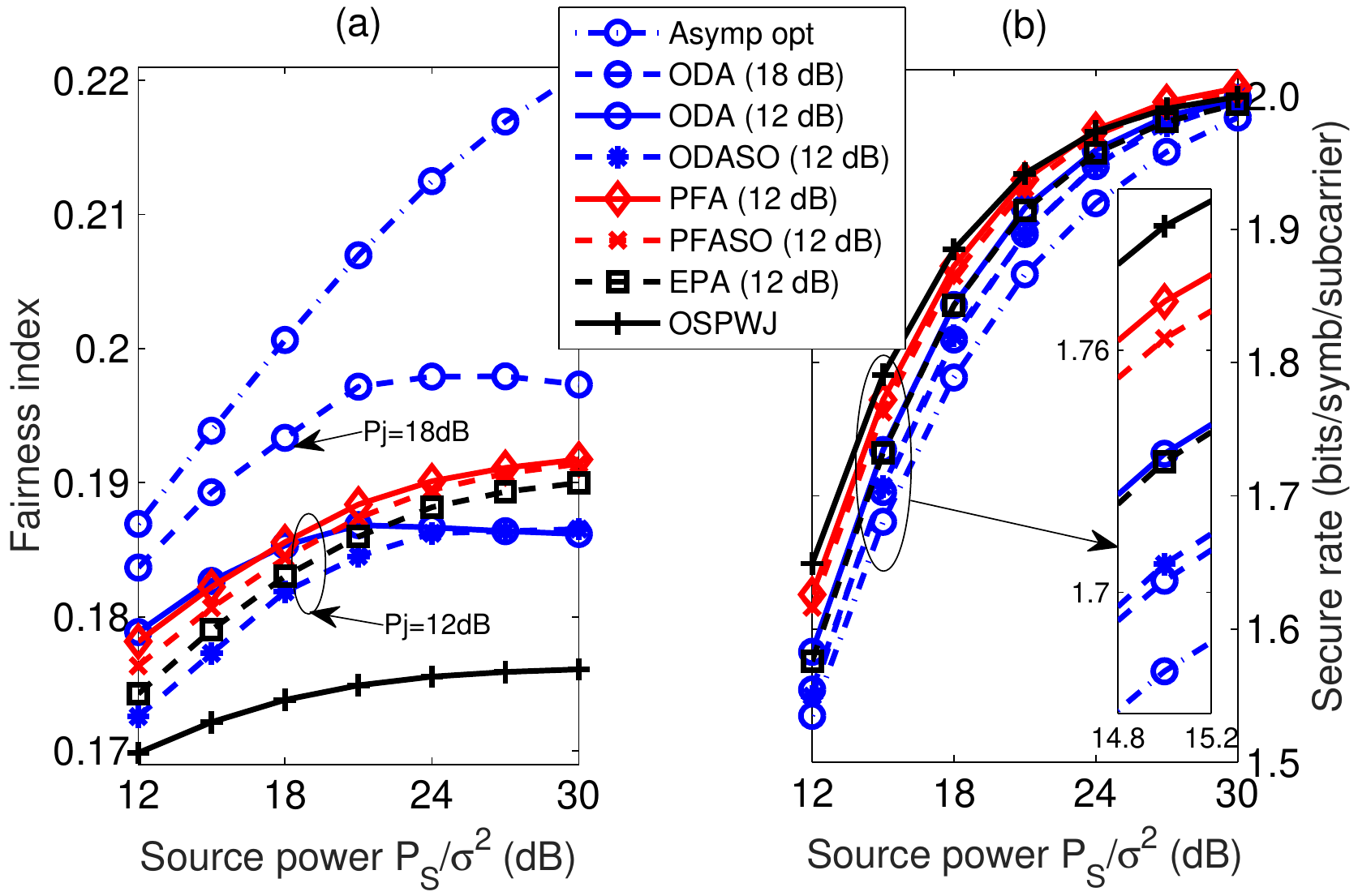,width=3.2in}
\caption{Fairness and secure rate versus source power at $P_{J1}/\sigma^2 = 12$ dB and $P_{J2}/\sigma^2 = 18$ dB.}
\label{fig:max_min_ps_variation}
\end{figure}

The performance of the max-min algorithms versus jammer power is presented in Fig.  \ref{fig:max_min_pj_variation}.
For fixed source power, performance of ODA as well as PFA keep improving with increasing jammer power, because of the increasing possibility of optimal jammer power allocation $(P_{j_n}^\star)$ over snatched subcarriers, and finally saturates.
At lower jammer power, performance of ODASO is poorer compared to EPA, because ODASO is able to help only a few initial users while EPA provides equal opportunity to all the users.
Because of the increasing possibility of allocating $P_{j_n}^\star$, ODASO outperforms EPA as well as PFASO at higher jammer power.
Both EPA and PFASO are unaware of the existence of optimal jammer power over snatched subcarrier. As a result, they utilize more than the required jammer power, leading to their degraded performance at higher jammer power. 
As observed in corresponding secure rate performance in \ref{fig:max_min_pj_variation}(b), secure rate of all the algorithms keep reducing  with jammer power due to increasing possibility of subcarrier snatching.
PFA and ODA saturates at higher jammer power because of the utilization of optimal jammer power $(P_{j_n}^\star)$ by the algorithms. 
ODASO saturates at a lower value compared to ODA because ODA uses power optimization while ODASO does not. 
Similarly PFASO shows poorer secure rate performance because of sequential power allocation.
The secure rate of EPA also saturates at higher jammer power, due to the existence of upper jammer power  bounds, but at relatively lower values compared to that of PFASO which uses better jammer power allocation (cf. (\ref{jammer_power_allocation_maxmin})).

\begin{figure}[!htb]
\centering
\epsfig{file=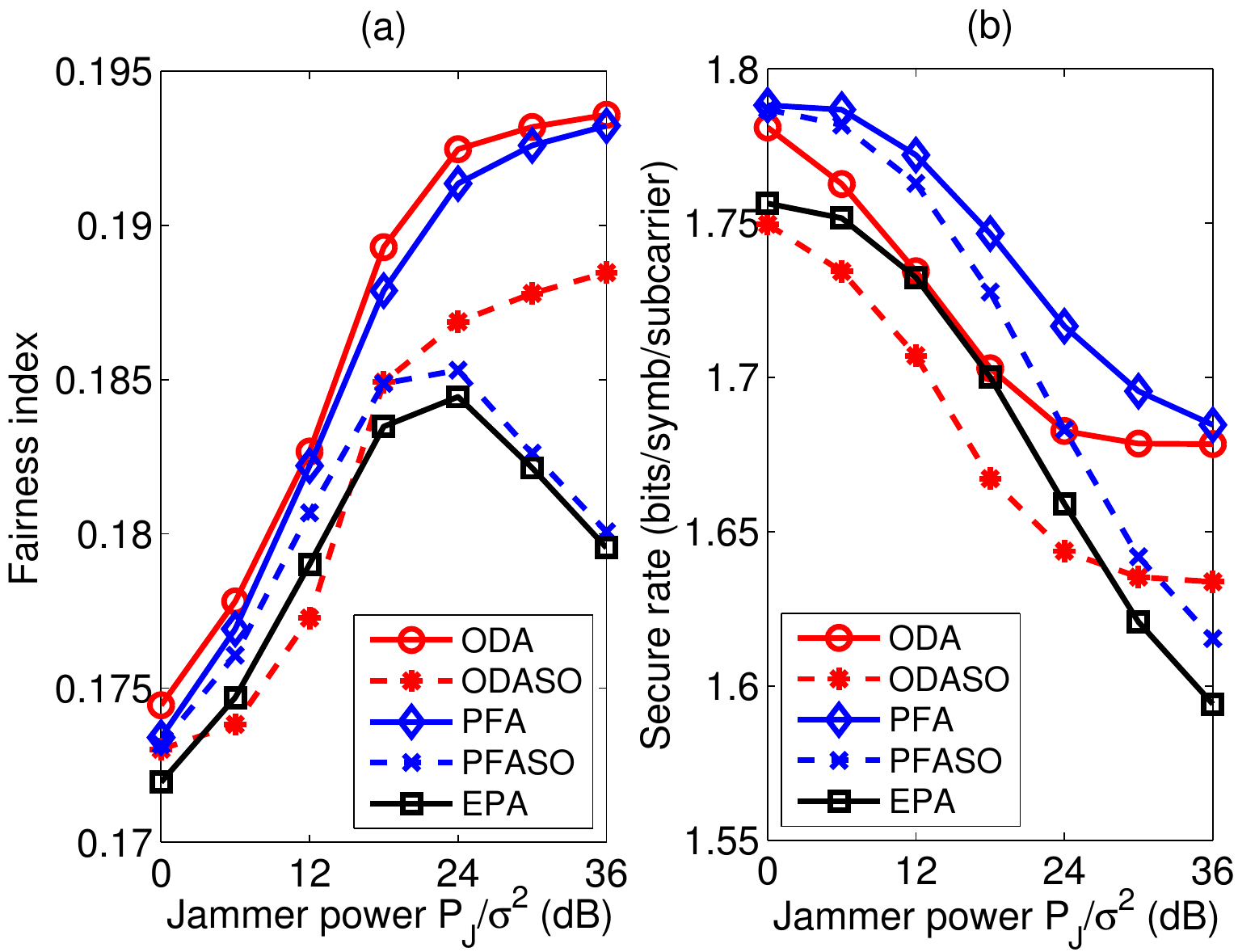,width=3.2in}
\caption{Fairness and secure rate versus jammer power at $P_{S}/\sigma^2 = 15$ dB.}
\label{fig:max_min_pj_variation}
\end{figure}

\section{Concluding Remarks} 	
\label{sec_conclusion}
We have considered jammer power utilization in the downlink of broadcast secure OFDMA for either secure rate or fairness improvement.
We have shown that secure rate can be improved under certain channel conditions with constraints on source and jammer power.
Also in presence of jammer, maximum secure rate is achieved at an optimal jammer power.
While solving the otherwise NP-hard resource allocation problem in steps, 
we have used PD and AO procedures for joint source and jammer power optimization.
It has been observed that sum secure rate can be significantly improved by using jammer power.
In max-min fairness scenario we have shown that, with the help of jammer a subcarrier can be snatched form its best gain user and allocated to a poor user who would otherwise have a low secure rate due to resource scarcity. 
To this end, we have presented two means for jammer power usage: PFA and ODA.
Overall performance of PFA is better. ODA performs  
marginally better at lower source power due to dynamic allocation of resources.
But at higher source power, with increasing user demands ODA performs poorly because of depleting  jammer resources for a few initial users. 
Asymptotically optimal solutions have also been derived to benchmark optimality of the proposed schemes. 

Possible future extensions 
include: consideration of multiple antenna at the nodes, imperfect CSI, finding the best location of jammer, and multiplicity of jammer. However, in each of the problems above the combinatorial aspect will open up new challenges. Imperfect CSI can lead to false decisions on subcarrier allocation, which may lead to zero secure rate. Though multi-antenna nodes or multiple jammers will add new degrees of freedom, the resource allocation algorithms need to be modified significantly. Optimizing jammer location is an interesting problem due to users' mobility. 

\appendices
\setcounter{equation}{0}
\setcounter{figure}{0}
\renewcommand{\theequation}{A.\arabic{equation}}
\renewcommand{\thefigure}{A.\arabic{figure}}

\section{Proof of proposition III.1}\label{sec_appendix_proposition_rate_imp}

Over a subcarrier $n$ with user $m$ and eavesdropper $e$, secure rate improvement implies that the secure rate with jammer should be greater than that without jammer, i.e.,

\begin{equation*}\label{inequality_secure_rate_imp}
\log_2 \left( \frac{1+\frac{P_{s_n} |h_{m,n}|^2}{\sigma^2 + P_{j_n} |g_{m,n}|^2}} {1+\frac{P_{s_n} |h_{e,n}|^2}{\sigma^2 + P_{j_n} |g_{e,n}|^2}} \right)-\log_2 \left( \frac{1+\frac{P_{s_n} |h_{m,n}|^2}{\sigma^2}}{1+\frac{P_{s_n} |h_{e,n}|^2}{\sigma^2}} \right) > 0 .
\end{equation*}
After simplifications we have following inequality in $P_{j_n}$:
\begin{align}\label{inequality_rate_improvement_p_j_n}
& P_{j_n} |g_{m,n}|^2 |g_{e,n}|^2 \left( |h_{m,n}|^2- |h_{e,n}|^2 \right) \nonumber \\
& \qquad <  P_{s_n} \left( |g_{e,n}|^2 - |g_{m,n}|^2 \right) |h_{m,n}|^2 |h_{e,n}|^2  \nonumber \\
& \qquad \qquad + \sigma^2 \left( |g_{e,n}|^2 |h_{e,n}|^2 - |g_{m,n}|^2 |h_{m,n}|^2 \right).
\end{align}

Since $|h_{m,n}|^2 > |h_{e,n}|^2$, coefficient of $P_{j_n}$ is positive. To have $P_{j_n}>0$, right hand side of \eqref{inequality_rate_improvement_p_j_n} should be positive, i.e.,  
\begin{align}\label{jammer_rate_imp_psn_constraint}
& P_{s_n} \left( |g_{e,n}|^2 - |g_{m,n}|^2 \right) |h_{m,n}|^2 |h_{e,n}|^2 \nonumber \\
& \quad + \sigma^2 \left( |g_{e,n}|^2 |h_{e,n}|^2 - |g_{m,n}|^2 |h_{m,n}|^2 \right) > 0.
\end{align}

If $|g_{m,n}|^2 > |g_{e,n}|^2$, then $|g_{m,n}|^2 |h_{m,n}|^2>|g_{e,n}|^2 |h_{e,n}|^2$, therefore, (\ref{jammer_rate_imp_psn_constraint}) cannot be satisfied for any $P_{s_n}>0$. Thus, when the jammer is affecting the intended user strongly compared to the eavesdropper, secure rate cannot be improved. If $|g_{m,n}|^2 < |g_{e,n}|^2$, the source power $P_{s_n}$ is conditioned as:
\begin{equation}\label{jammer_rate_psn_inequality}
P_{s_n} > \frac{ \sigma^2 \left( |g_{m,n}|^2 |h_{m,n}|^2 - |g_{e,n}|^2 |h_{e,n}|^2 \right) }{ \left( |g_{e,n}|^2 - |g_{m,n}|^2 \right) |h_{m,n}|^2 |h_{e,n}|^2  } \triangleq P_{s_{n}}^{th_i}.
\end{equation}

If $|g_{m,n}|^2 |h_{m,n}|^2 - |g_{e,n}|^2 |h_{e,n}|^2>0$, $P_{s_n}$ should be above certain threshold $P_{s_{n}}^{th_i}$ as mentioned in \eqref{jammer_rate_psn_inequality}.  
If $|g_{m,n}|^2 |h_{m,n}|^2 - |g_{e,n}|^2 |h_{e,n}|^2<0$,  $P_{s_{n}}^{th_i}$ becomes negative, which means that the secure rate can be improved for any $P_{s_n}>0$.
Thus, under certain channel conditions and certain source power and jammer power constraints the secure rate of a user can be improved over a subcarrier.
Note that, in order to have positive secure rate improvement, the jammer power $P_{j_{n}}$ should be below certain threshold $P_{j_{n}}^{th_i}$ as described in (\ref{inequality_rate_improvement_p_j_n}).
The secure rate with jammer is equal to the secure rate without jammer at $P_{j_n} = 0$ and $P_{j_n} = P_{j_{n}}^{th_i}$, and in between these two extremes there is a secure rate improvement. Intuitively it indicates the presence of a maxima as described below.

From (\ref{secure_rate_definition}) first derivative of  $R_{m,n}$ with respect to $P_{j_n}$ is:
\begin{equation*}
\frac{\delta R_{m,n}}{\delta P_{j_n}} = \frac{1}{\ln2} \left[\frac{ \frac{-P_{s_n} |h_{m,n}|^2 |g_{m,n}|^2}{\left(\sigma^2 + P_{j_n}|g_{m,n}|^2\right)^2}} {1 + \frac{P_{s_n}|h_{m,n}|^2}{\sigma^2 + P_{j_n}|g_{m,n}|^2}} 
 - \frac{ \frac{-P_{s_n} |h_{e,n}|^2 |g_{e,n}|^2}{\left(\sigma^2 + P_{j_n}|g_{e,n}|^2\right)^2}} {1 + \frac{P_{s_n}|h_{e,n}|^2}{\sigma^2 + P_{j_n}|g_{e,n}|^2}} \right].
\end{equation*}
Simplifying it further we get,
\begin{align*}
&\frac{\delta R_{m,n}}{\delta P_{j_n}} = \left[ \frac{P_{s_n} |h_{e,n}|^2 |g_{e,n}|^2}{\left( \sigma^2 + P_{j_n}|g_{e,n}|^2 \right) \left( \sigma^2 + P_{j_n}|g_{e,n}|^2 + P_{s_n}|h_{e,n}|^2 \right) } \right. \nonumber \\
&- \left. \frac{P_{s_n} |h_{m,n}|^2 |g_{m,n}|^2}{ \left( \sigma^2 + P_{j_n}|g_{m,n}|^2 \right) \left( \sigma^2 + P_{j_n}|g_{m,n}|^2 + P_{s_n}|h_{m,n}|^2 \right)} \right] \frac{1}{\ln2}.
\end{align*}

The derivative has a quadratic in numerator and a fourth-order equation of $P_{j_n}$ with all positive coefficients in the denominator. Hence, the denominator is always positive when $P_{j_n}>0$.
Setting the derivative equal to zero, we obtain a quadratic equation of the form $x_n P_{j_n}^2 + y_n P_{j_n} + z_n = 0 $ with the following coefficients: 
\begin{align}\label{rate_derivative_pjn_quad_coef1}
x_n &= |g_{m,n}|^2 |g_{e,n}|^2 \left( |g_{m,n}|^2 |h_{e,n}|^2-|g_{e,n}|^2 |h_{m,n}|^2 \right) \\
y_n &= 2 \sigma^2 |g_{m,n}|^2 |g_{e,n}|^2 \left( |h_{e,n}|^2-|h_{m,n}|^2 \right)  
\label{rate_derivative_pjn_quad_coef2}\\
z_n & = \sigma^2 P_{s_n} |h_{m,n}|^2 |h_{e,n}|^2 \left( |g_{e,n}|^2-|g_{m,n}|^2 \right) \nonumber \\
& \quad + \sigma^4 \left( |g_{e,n}|^2 |h_{e,n}|^2-|g_{m,n}|^2 |h_{m,n}|^2 \right) 
\label{rate_derivative_pjn_quad_coef3}
\end{align}
The rate improvement scenario, i.e., $|h_{m,n}|^2>|h_{e,n}|^2$ and $|g_{e,n}|^2>|g_{m,n}|^2$, results in $x_n<0$, $y_n<0$,
and 
$z_n>0$.
Since the discriminant $\Delta_n = \sqrt{y_n^2-4x_nz_n}$ is positive, and $x_n$ and $z_n$ have opposite signs, there exists only one positive real root of the above quadratic equation. 
Thus, the derivative $\frac{\delta R_{m,n}}{\delta P_{j_n}}$ has only one zero crossing in $P_{j_n}>0$ which corresponds to a unique maxima of rate $R_{m,n}$ with respect to jammer power $P_{j_n}$.
The optimum jammer power $(P_{j_n}^{o})$ achieving maximum secure rate is obtained as the positive real root of the above quadratic equation. Observing (\ref{rate_derivative_pjn_quad_coef1}) to (\ref{rate_derivative_pjn_quad_coef3}), we note that $P_{j_n}^{o}$ is an increasing function of $P_{s_n}$

\setcounter{equation}{0}
\setcounter{figure}{0}
\renewcommand{\theequation}{B.\arabic{equation}}
\renewcommand{\thefigure}{B.\arabic{figure}}
\section{Proof of proposition IV.1}\label{sec_appendix_proposition_carr_snatch}
The condition of subcarrier snatching for user $m$ from user $e$ is to achieve positive secure rate over subcarrier $n$, i.e., 
\begin{align}\label{eqn_mmf_snatch_condition}
\log_2 \left( 1+ \gamma_{m,n}' \right) - \log_2 \left( 1+\gamma_{e,n}' \right)  >0
\end{align}
Simplifying above equation we have
\begin{equation}
\frac{P_{s_n}|h_{m,n}|^2}{\sigma^2+P_{j_n}|g_{m,n}|^2} > \frac{P_{s_n}|h_{e,n}|^2}{\sigma^2+P_{j_n}|g_{e,n}|^2}
\end{equation}
which gives the following constraint over jammer power $P_{j_n}$
\begin{align}\label{eqn_mmf_jammer_power_for_carr_snatching}
& P_{j_n} \left( |g_{e,n}|^2 |h_{m,n}|^2 - |g_{m,n}|^2 |h_{e,n}|^2 \right) \nonumber \\
& \qquad  > \sigma^2 \left( |h_{e,n}|^2 - |h_{m,n}|^2 \right).
\end{align}
Since $|h_{m,n}|^2 < |h_{e,n}|^2$, if $|g_{e,n}|^2>|g_{m,n}|^2$ such that $|g_{e,n}|^2 |h_{m,n}|^2 > |g_{m,n}|^2 |h_{e,n}|^2 $, then left hand side of (\ref{eqn_mmf_jammer_power_for_carr_snatching}) is positive and the jammer power is constrained as
\begin{equation}
 P_{j_n} > \frac {\sigma^2 \left( |h_{e,n}|^2 - |h_{m,n}|^2 \right)} {|g_{e,n}|^2 |h_{m,n}|^2 - |g_{m,n}|^2 |h_{e,n}|^2} \triangleq P_{j_{n}}^{th_s}.
\end{equation}
In case $|g_{e,n}|^2 |h_{m,n}|^2 < |g_{m,n}|^2 |h_{e,n}|^2 $, then the left hand side in (\ref{eqn_mmf_jammer_power_for_carr_snatching}) becomes negative.
Since the inequality cannot be satisfied for any positive jammer power, the jammer cannot be utilized to snatch the subcarrier in such channel conditions.
Further, if $|g_{e,n}|^2 < |g_{m,n}|^2$, then also $|g_{e,n}|^2 |h_{m,n}|^2 < |g_{m,n}|^2 |h_{e,n}|^2$, and thus using the same argument as above,
jammer cannot help in snatching the subcarrier.

Let us refer to the equations (\ref{rate_derivative_pjn_quad_coef1}) to (\ref{rate_derivative_pjn_quad_coef3}), which are the coefficients of the quadratic equation obtained after setting $\frac{\delta R_{m,n}}{\delta P_{j_n}}=0$. 
Under the subcarrier snatching scenario, we have $|h_{m,n}|^2 < |h_{e,n}|^2$, $|g_{e,n}|^2>|g_{m,n}|^2$, and $|g_{e,n}|^2 |h_{m,n}|^2 > |g_{m,n}|^2 |h_{e,n}|^2$, which cause $x_n<0, y_n>0$, and $z_n>0$.
These conditions indicate the existence of a positive real root, which corresponds to the optimal jammer power $P_{j_n}^{o}$ required to achieve maximum secure rate over the snatched subcarrier.


\section*{Acknowledgment}The authors are thankful to the anonymous reviewers for the insightful comments and valuable suggestions, which have significantly improved the quality of presentation.

\begin{IEEEbiography}
[{\includegraphics[width=1in]{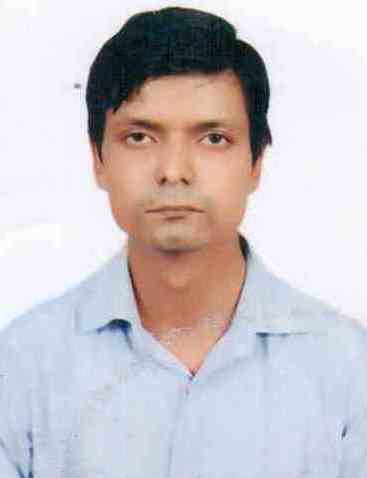}}]{Ravikant Saini [S'12]} received the B.Tech. degree in Electronics and Communication Engineering and M.Tech. degree in Communication Systems from the Indian Institute of Technology Roorkee, India in 2001 and 2005, respectively. He is currently working towards his Ph.D. degree at the Indian Institute of Technology Delhi, India. His research interests include wireless communication, resource allocation and physical layer security. 
\end{IEEEbiography}

\begin{IEEEbiography}
[{\includegraphics[width=1in]{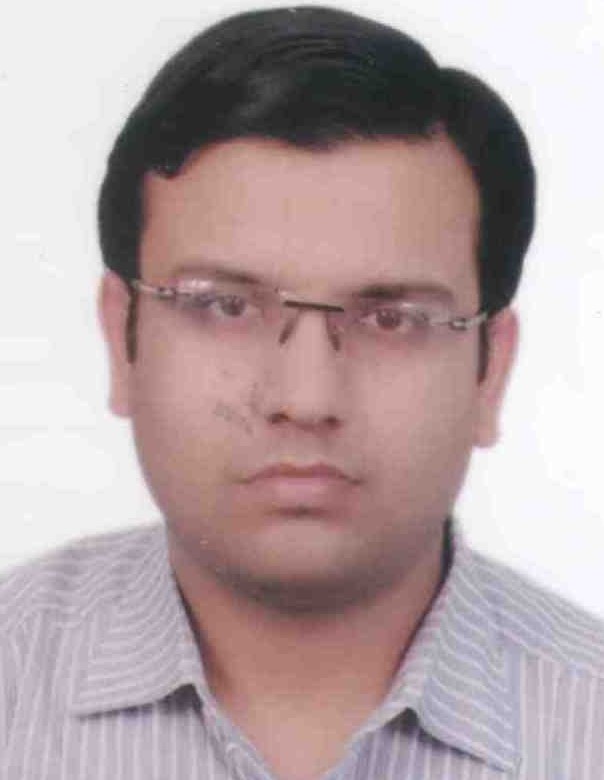}}]{Abhishek Jindal [S'08]} received the B.Tech and M.Tech degrees from the Jaypee Institute of Information Technology, Noida, India in 2009 and 2011 respectively. He is currently working towards his Ph.D. degree at the Indian Institute of Technology Delhi, India. His research interests include performance study and resource allocation for physical layer security.
\end{IEEEbiography}

\begin{IEEEbiography}
[{\includegraphics[width=1in]{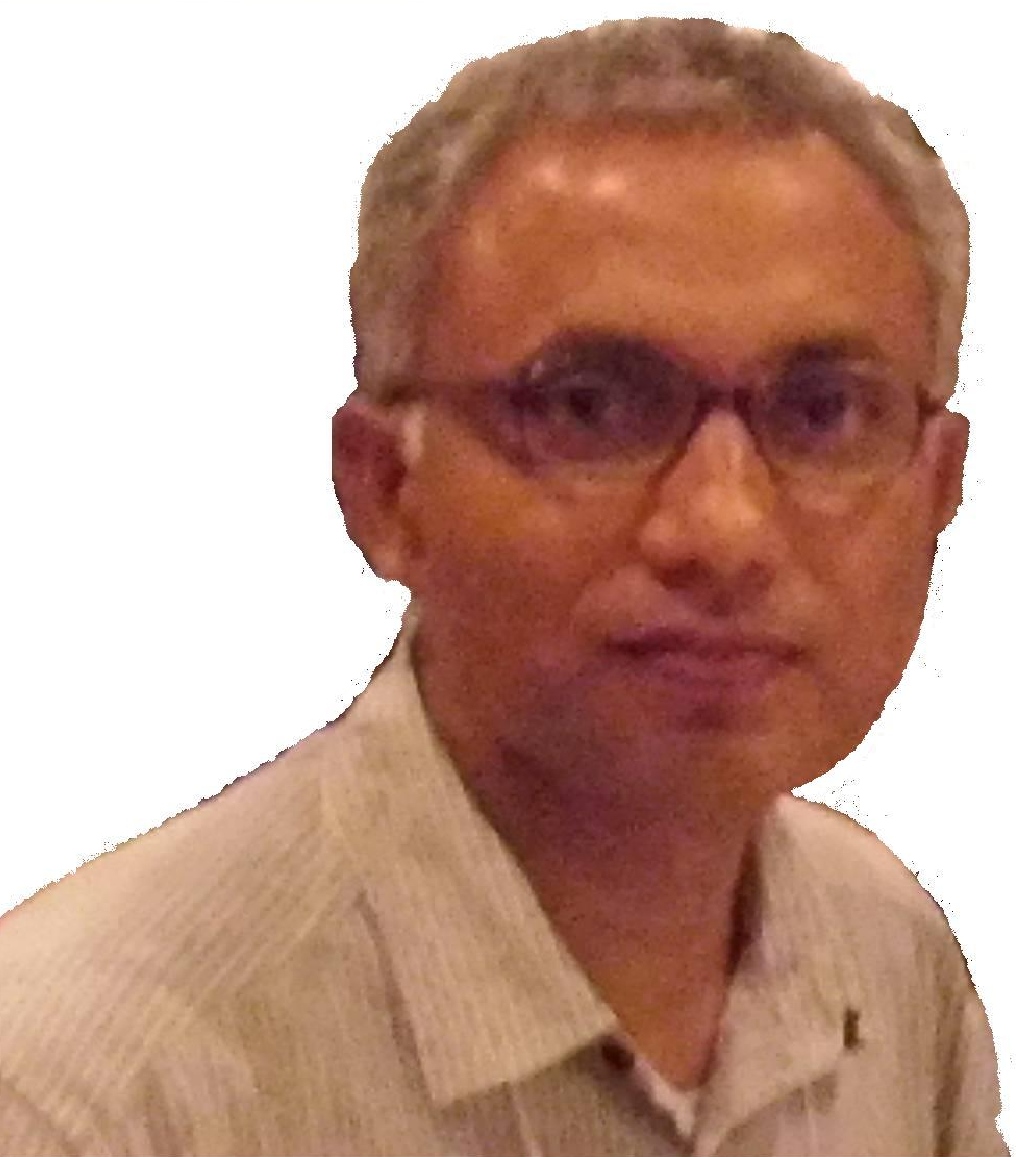}}]{Swades De [S'02-M'04-SM'14]} received the Ph.D. degree from State University of New York at Buffalo, NY, USA, in 2004. He is currently an Associate Professor in the Department of Electrical Engineering, IIT Delhi, India. In 2004, he worked as an ERCIM researcher at ISTI-CNR, Italy. From 2004 to 2006 he was with NJIT, NJ, USA as an Assistant Professor. His research interests include performance study, resource efficiency in wireless networks, broadband access networks and systems.
\end{IEEEbiography}


\begin{thebibliography}{10}
\providecommand{\url}[1]{#1}
\csname url@samestyle\endcsname
\providecommand{\newblock}{\relax}
\providecommand{\bibinfo}[2]{#2}
\providecommand{\BIBentrySTDinterwordspacing}{\spaceskip=0pt\relax}
\providecommand{\BIBentryALTinterwordstretchfactor}{4}
\providecommand{\BIBentryALTinterwordspacing}{\spaceskip=\fontdimen2\font plus
\BIBentryALTinterwordstretchfactor\fontdimen3\font minus
  \fontdimen4\font\relax}
\providecommand{\BIBforeignlanguage}[2]{{%
\expandafter\ifx\csname l@#1\endcsname\relax
\typeout{** WARNING: IEEEtran.bst: No hyphenation pattern has been}%
\typeout{** loaded for the language `#1'. Using the pattern for}%
\typeout{** the default language instead.}%
\else
\language=\csname l@#1\endcsname
\fi
#2}}
\providecommand{\BIBdecl}{\relax}
\BIBdecl

\bibitem{amitav_TCST_2014}
A.~Mukherjee, S.~Fakoorian, J.~Huang, and A.~Swindlehurst, ``Principles of
  physical layer security in multiuser wireless networks: A survey,''
  \emph{{IEEE} Commun. Surveys Tuts.}, vol.~16, no.~3, pp. 1550--1573, Aug.
  2014.

\bibitem{Shiu_WC_2011}
Y.-S. Shiu, S.~Y. Chang, H.-C. Wu, S.-H. Huang, and H.-H. Chen, ``Physical
  layer security in wireless networks: a tutorial,'' \emph{{IEEE} Wireless
  Commun.}, vol.~18, no.~2, pp. 66--74, Apr. 2011.

\bibitem{Bloch_book_2011}
M.~Bloch and J.~Barros, \emph{Physical-Layer Security: From Information Theory
  to Security Engineering}.\hskip 1em plus 0.5em minus 0.4em\relax Cambridge
  University Press, 2011.

\bibitem{Haohao_ICC_2013}
H.~Qin, X.~Chen, X.~Zhong, F.~He, M.~Zhao, and J.~Wang, ``Joint power
  allocation and artificial noise design for multiuser wiretap {OFDM}
  channels,'' in \emph{Proc. IEEE ICC}, Budapest, Hungary, June 2013.

\bibitem{Karachontzitis_TIFS_2015}
S.~Karachontzitis, S.~Timotheou, I.~Krikidis, and K.~Berberidis,
  ``Security-aware max-min resource allocation in multiuser {OFDMA} downlink,''
  \emph{{IEEE} Trans. Inf. Forensics Security}, vol.~10, no.~3, pp. 529--542,
  Mar. 2015.

\bibitem{Derrick_TVT_2012}
D.~Ng, E.~Lo, and R.~Schober, ``Energy-efficient resource allocation for secure
  {OFDMA} systems,'' \emph{{IEEE} Trans. Veh. Technol.}, vol.~61, no.~6, pp.
  2572--2585, July 2012.

\bibitem{Derrick_TWC_2011}
------, ``Secure resource allocation and scheduling for {OFDMA}
  decode-and-forward relay networks,'' \emph{{IEEE} Trans. Wireless Commun.},
  vol.~10, no.~10, pp. 3528--3540, Oct. 2011.

\bibitem{Munnujahan_ICC_2013}
M.~Ara, H.~Reboredo, F.~Renna, and M.~Rodrigues, ``Power allocation strategies
  for {OFDM} gaussian wiretap channels with a friendly jammer,'' in \emph{Proc.
  IEEE ICC}, Budapest, Hungary, June 2013, pp. 3413--3417.

\bibitem{Wang_eurasip_2013}
A.~Wang, J.~Chen, Y.~Cai, C.~Cai, W.~Yang, and Y.~Cheng, ``Joint subcarrier and
  power allocation for physical layer security in cooperative {OFDMA}
  networks,'' \emph{EURASIP J. wirel. commun. netw.}, vol. 2013, no.~1, pp.
  193--202, July 2013.

\bibitem{Jorswieck2008}
E.~Jorswieck and A.~Wolf, ``Resource allocation for the wire-tap multi-carrier
  broadcast channel,'' in \emph{Proc. ICT}, St. Petersburg, Russia, June 2008.

\bibitem{Xiaowei_TIFS_2011}
X.~Wang, M.~Tao, J.~Mo, and Y.~Xu, ``Power and subcarrier allocation for
  physical-layer security in {OFDMA}-based broadband wireless networks,''
  \emph{{IEEE} Trans. Inf. Forensics Security}, vol.~6, no.~3, pp. 693--702,
  Sep. 2011.

\bibitem{Jorswieck_SAC_2013}
Z.~Ho, E.~Jorswieck, and S.~Engelmann, ``Information leakage neutralization for
  the multi-antenna non-regenerative relay-assisted multi-carrier interference
  channel,'' \emph{{IEEE} J. Sel. Areas Commun.}, vol.~31, no.~9, pp.
  1672--1686, Sep. 2013.

\bibitem{Derrick_TWC_2015}
C.~Wang, H.-M. Wang, D.~Ng, X.~gen Xia, and C.~Liu, ``Joint beamforming and
  power allocation for secrecy in peer-to-peer relay networks,'' \emph{{IEEE}
  Trans. Wireless Commun.}, vol.~14, no.~6, pp. 3280--3293, June 2015.

\bibitem{Derrick_TVT_2015}
D.~Ng, E.~Lo, and R.~Schober, ``Multi-objective resource allocation for secure
  communication in cognitive radio networks with wireless information and power
  transfer,'' \emph{{IEEE} Trans. Veh. Technol.}, vol.~PP, no.~99, pp. 1--1,
  2015.

\bibitem{Yang_TC_2014}
N.~Yang, G.~Geraci, J.~Yuan, and R.~Malaney, ``Confidential broadcasting via
  linear precoding in non-homogeneous {MIMO} multiuser networks,'' \emph{{IEEE}
  Trans. Commun.}, vol.~62, no.~7, pp. 2515--2530, July 2014.

\bibitem{Geraci_TC_2012}
G.~Geraci, M.~Egan, J.~Yuan, A.~Razi, and I.~Collings, ``Secrecy sum-rates for
  multi-user {MIMO} regularized channel inversion precoding,'' \emph{{IEEE}
  Trans. Commun.}, vol.~60, no.~11, pp. 3472--3482, Nov. 2012.

\bibitem{SGOEL_TWC_2008}
S.~Goel and R.~Negi, ``Guaranteeing secrecy using artificial noise,''
  \emph{{IEEE} Trans. Wireless Commun.}, vol.~7, no.~6, pp. 2180--2189, June
  2008.

\bibitem{XTANG_TIT_2011}
X.~Tang, R.~Liu, P.~Spasojevic, and H.~Poor, ``Interference assisted secret
  communication,'' \emph{{IEEE} Trans. Inf. Theory}, vol.~57, no.~5, pp.
  3153--3167, May 2011.

\bibitem{Di_Yuan_TVT_2013}
D.~Yuan, J.~Joung, C.~K. Ho, and S.~Sun, ``On tractability aspects of optimal
  resource allocation in {OFDMA} systems,'' \emph{{IEEE} Trans. Veh. Technol.},
  vol.~62, no.~2, pp. 863--873, Feb. 2013.

\bibitem{Y_F_LIU_TSP_2014}
Y.-F. Liu and Y.-H. Dai, ``On the complexity of joint subcarrier and power
  allocation for multi-user {OFDMA} systems,'' \emph{{IEEE} Trans. Signal
  Process.}, vol.~62, no.~3, pp. 583--596, Feb. 2014.

\bibitem{MTAO_TWC_2008}
M.~Tao, Y.-C. Liang, and F.~Zhang, ``Resource allocation for delay
  differentiated traffic in multiuser {OFDM} systems,'' \emph{{IEEE} Trans.
  Wireless Commun.}, vol.~7, no.~6, pp. 2190--2201, June 2008.

\bibitem{ravikant_ICC_2015}
R.~Saini, A.~Jindal, and S.~De, ``Jammer assisted sum rate and fairness
  improvement in secure ofdma,'' in \emph{Proc. IEEE ICC}, London, UK, June
  2015.

\bibitem{sboyd_decomposition}
\BIBentryALTinterwordspacing
S.~Boyd, L.~Xiao, A.~Mutapcic, and J.~Mattingley. (2008, Apr.) Notes on
  decomposition methods. [Online]. Available:
  \url{http://see.stanford.edu/materials/lsocoee364b/08-decomposition\_notes.pdf}
\BIBentrySTDinterwordspacing

\bibitem{sboyd_alternating}
\BIBentryALTinterwordspacing
S.~Boyd. Sequential convex programming. [Online]. Available:
  \url{http://stanford.edu/class/ee364b/lectures/seq\_slides.pdf}
\BIBentrySTDinterwordspacing

\bibitem{sboyd_subgradient}
\BIBentryALTinterwordspacing
S.~Boyd, L.~Xiao, and A.~Mutapcic. (2008, Apr.) Subgradient methods. [Online].
  Available:
  \url{http://see.stanford.edu/materials/lsocoee364b/02-subgrad\_method\_notes.pdf}
\BIBentrySTDinterwordspacing

\bibitem{HMWANG_TIFS_2014}
H.-M. Wang, F.~Liu, and X.-G. Xia, ``Joint source-relay precoding and power
  allocation for secure amplify-and-forward {MIMO} relay networks,''
  \emph{{IEEE} Trans. Inf. Forensics Security}, vol.~9, no.~8, pp. 1240--1250,
  Aug. 2014.

\bibitem{Grippo_ORL_2000}
L.~Grippo and M.~Sciandrone, ``On the convergence of the block nonlinear
  gaussâ€“seidel method under convex constraints,'' \emph{Operations Research
  Letters}, vol.~26, no.~3, pp. 127 -- 136, Apr. 2000.

\bibitem{W_Yu_TCOM_2006}
W.~Yu and R.~Lui, ``Dual methods for nonconvex spectrum optimization of
  multicarrier systems,'' \emph{{IEEE} Trans. Commun.}, vol.~54, no.~7, pp.
  1310--1322, July 2006.

\bibitem{Rhee2000}
W.~Rhee and J.~M. Cioffi, ``Increase in capacity of multiuser {OFDM} system
  using dynamic subchannel allocation,'' in \emph{Proc. {IEEE} VTC-{S}pring},
  Tokyo, Japan, May 2000, pp. 1085--89.

\bibitem{Jang2003}
J.~Jang and K.~B. Lee, ``Transmit power adaptation for multiuser {OFDM}
  systems,'' \emph{{IEEE} J. Sel. Areas Commun.}, vol.~21, no.~2, pp. 171--178,
  Feb. 2003.

\bibitem{Altman_alpha_2008}
E.~Altman, K.~Avrachenkov, and A.~Garnaev, ``Generalized $\alpha$-fair resource
  allocation in wireless networks,'' in \emph{Decision and Control, 2008. CDC
  2008. 47th IEEE Conference on}, Cancun, Mexico, Dec 2008, pp. 2414--2419.

\bibitem{Mohanram_CL_2005}
C.~Mohanram and S.~Bhashyam, ``A sub-optimal joint subcarrier and power
  allocation algorithm for multiuser {OFDM},'' \emph{{IEEE} Commun. Lett.},
  vol.~9, no.~8, pp. 685--87, Aug. 2005.

\bibitem{jain1984}
R.~Jain, D.-M. Chiu, and W.~R. Hawe, ``A quantitative measure of fairness and
  discrimination for resource allocation in shared computer system,'' Digital
  Equipment Corporation, Tech. Rep., 1984.

\end{thebibliography}
\end{document}